\documentclass{amsart}
\usepackage{amsmath, amsthm}
\usepackage{amssymb, amscd, latexsym, graphicx, psfrag}
\usepackage{hyperref}
\hypersetup{colorlinks=true, linkcolor=black}
\numberwithin{equation}{section}
\newtheorem{Prop}{Proposition}
\newtheorem{Thm}{Theorem}
\newtheorem{Lemma}{Lemma}
\newtheorem{Cor}{Corollary}
\theoremstyle{definition}
\newtheorem{Def}{Definition}
\theoremstyle{remark}
\newtheorem{Rem}{Remark}

\theoremstyle{definition}

\title{4-dimensional Frobenius manifolds and Painleve' \nolinebreak VI}
\author{Stefano Romano}
\date{}
\begin{document}
\begin{abstract} A Frobenius manifold has tri-hamiltonian structure if it is even-dimensional and its spectrum is maximally degenerate. We focus on the case of dimension four and show that, under the assumption of semisimplicity, the corresponding isomonodromic Fuchsian system is described by the Painlev\'e VI$\mu$ equation.  This yields an explicit procedure associating to any semisimple Frobenius manifold of dimension three a tri-hamiltonian Frobenius manifold of dimension four. We carry out explicit examples for the case of Frobenius structures on Hurwitz spaces.
\end{abstract}
\maketitle
\section*{Introduction}
The notion of Frobenius manifold was introduced by B. Dubrovin (\cite{Dub2Dtft}, \cite{DubPainP}, \cite{DubGAT}) as a coordinate-free formulation of the Witten-Dijkgraaf-Verlinde-Verlinde equations of associativity, whose prominent role in different branches of mathematics and theoretical physics (ranging from integrable systems and singularity theory to topological field theory and Gromov-Witten invariants) marked a major trend of research in the last over twenty years. By definition a function $F=F(t^1, \dots, t^n)$ of $n$ complex variables satisfies WDVV if the third derivatives
\begin{equation}
c_{\alpha\beta\gamma}(t)\doteq\frac{\partial^3F(t)}{\partial t^\alpha\partial t^\beta\partial t^\gamma}
\end{equation}
define an associative product of vector fields via the rule
\begin{equation}\label{product}
\frac{\partial}{\partial t^\alpha}\cdot\frac{\partial}{\partial t^\beta}\doteq c_{\alpha\beta}^\gamma(t)\frac{\partial}{\partial t^\gamma},\qquad c_{\alpha\beta}^\gamma(t)\doteq\eta^{\gamma\mu}c_{\alpha\beta\mu}(t)
\end{equation}
where $\eta^{\alpha\beta}$ is a fixed constant non-degenerate symmetric matrix. More precisely, Frobenius manifolds correspond to a particular class of solutions of WDVV, satisfying two constraints:
\begin{itemize}
\item[(A)] $\partial_1\partial_\alpha\partial_\beta F=\eta_{\alpha\beta}$, where $\eta_{\alpha\beta}$ is the inverse of $\eta^{\alpha\beta}$. This means that $\partial/\partial t^1$ is the unit of the product \eqref{product}.
\item[(B)]$\partial_E F =(3-d)F + \text{quadratic}$, where
\begin{equation}\label{Euler}
\partial_E=\sum_{\alpha=1}^n(d_\alpha t^\alpha+r_\alpha)\partial_\alpha
\end{equation}
for some constants $d, d_\alpha, r_\alpha$ with $d_1=1$, $r_\alpha=0$ if $d_\alpha\neq 0$. This means that $F$ is quasi-homogeneous up to quadratic terms. $\partial_E$ is called the \emph{Euler vector field} and $d$ is called the \emph{charge}.
\end{itemize}
In the geometric picture, $\eta$ is a flat metric on the Frobenius manifold with flat coordinates $t^\alpha$, and the product \eqref{product} equips the tangent spaces with the structure of a Frobenius algebra (see \cite{Dub2Dtft}, \cite{DubPainP} for details). The function $F$ is called \emph{prepotential}, or \emph{primary free energy}.\\
A fundamental consequence of the quasi-homogeneity axiom (B) is the existence of a \emph{flat pencil} of metrics on the manifold. More precisely, let 
$$
\mathcal{U}^\alpha_\beta\doteq(\partial_E\; \cdot)^\alpha_\beta=(\partial_E)^\mu c^\alpha_{\mu\beta}
$$
denote the operator of multiplication by the Euler vector field; then the \emph{intersection form}
\begin{equation}\label{intform}
g^{\alpha\beta}\doteq \eta^{\alpha\mu}\mathcal{U}^\beta_\mu
\end{equation}
is flat and compatible with $\eta$ (\cite{DubFlat}), leading to a bi-hamiltonian structure of hydrodynamic type on the formal loop space with target the Frobenius manifold (\cite{Dub-Nov}, \cite{Dub-Nov2}).\\
\indent In this work we consider a particular class of Frobenius manifolds, singled out by the following additional condition:
\begin{itemize}
\item[(C)] The dimension is even, $n=2k$, and the Euler vector field is of the form
$$
\partial_E=\sum_{\alpha=1}^kt^\alpha\frac{\partial}{\partial t^\alpha}+(1+2\mu)\sum_{\alpha=k+1}^nt^\alpha\frac{\partial}{\partial t^\alpha}
$$
for some non-zero constant $\mu$.
\end{itemize}
\begin{Def} A Frobenius manifold satisfying (C) is said to have \emph{tri-hamiltonian structure}.
\end{Def}
The definition is motivated by the following
\begin{Lemma} If condition (C) is satisfied, the metric
\begin{equation}\label{tildeeta}
\tilde{\eta}^{\alpha\beta}\doteq \eta^{\alpha\mu}(\mathcal{U}^2)^\beta_\mu
\end{equation}
is flat and compatible with both $\eta$ and $g$.
\end{Lemma}
The result was first observed by Pavlov and Tsarev in the semisimple setting (\cite{PT}). Here we prove it in general; as it turns out, it actually holds under assumptions weaker than (C) in the non-semisimple case.
\vspace{5pt}\\
\indent The main goal of this paper is to give a complete description of semisimple tri-hamiltonian Frobenius manifolds in the lowest non-trivial dimension $n=4$; such objects are naturally identified with isomonodromic families of Fuchsian linear operators of the form
\begin{equation}\label{tri-system}
\frac{d}{d\epsilon}+\sum_{i=1}^4\frac{E_iW}{\epsilon-v_i}
\end{equation}
where $(E_i)_{jk}=\delta_{ji}\delta_{ik}$ and $W$ is a 4x4 matrix satisfying
\begin{equation}\label{tri-system2}
W^t=-W, \qquad W^2=\mu^2 I
\end{equation}
Our basic observation is the following:
\begin{Thm} The equations of isomonodromic deformation of \eqref{tri-system}, \eqref{tri-system2} are equivalent to the Painlev\'e VI$\mu$ equation (\cite{Dub2Dtft}, \cite{Dub-Maz})
\begin{align}\label{PVImu}
\frac{d^2y}{ds^2}=	&\frac{1}{2}\left(\frac{1}{y}+\frac{1}{y-1}+\frac{1}{y-s}\right)\left(\frac{dy}{ds}\right)^2-\left(\frac{1}{s}+\frac{1}{s-1}+\frac{1}{y-s}\right)\frac{dy}{ds}+\\
				&+\frac{1}{2}\frac{y(y-1)(y-s)}{s^2(s-1)^2}\left((2\mu-1)^2+\frac{s(s-1)}{y-s}^2\right)\nonumber
\end{align}
\end{Thm}
It was proved in (\cite{Dub2Dtft}, Appendix E) that the solutions of \eqref{PVImu} parametrize semisimple Frobenius manifolds of dimension 3. Thus the above results can be rephrased as a correspondence between Frobenius manifolds of dimension 3 and 4, the latter belonging to the tri-hamiltonian subclass. We proceed to derive a full construction associating to any Painlev\'e VI$\mu$ transcendent a tri-hamiltonian prepotential in four variables. Explicitly:
\begin{itemize}
\item Start with a 3-dimensional semisimple Frobenius manifold $M$, or equivalently with a Painvlev\'e VI$\mu$ transcendent.
\item Let $\{u_i\}_{i=1, 2, 3}$ be the canonical coordinates of $M$ and $\Psi$ be the \emph{transition matrix}
\begin{equation}\label{transitionPsi}
\Psi=(\psi_{i\alpha}),\qquad \psi_{i\alpha}=\sqrt{\eta_i}\frac{\partial u_i}{\partial t^\alpha}
\end{equation}
where $\eta_i=\eta(\partial/\partial u_i, \partial/\partial u_i)$. Define
$$
\Phi(s)\doteq \Psi(u_1=0, u_2=1, u_3=s)= \Psi(u) (u_2-u_1)^{-\hat{\mu}}
$$
where $\hat{\mu}=\text{diag}(\mu, 0, -\mu)$ and the identification $s=(u_3-u_1)/(u_2-u_1)$ is understood in the second equality.
\item Let $\{p_1, p_2, t^3\}$ be a basis of twisted periods of $M$ for the value $\nu=\mu+1/2$ of the twisting parameter. Define
$$
\chi_{1i}(\epsilon, s)\doteq\frac{\partial p_i}{\partial t^2}\qquad \chi_{2i}(\epsilon, s)\doteq\frac{\partial p_i}{\partial t^1}\qquad\;\;\; i=1, 2
$$
where the right hand sides are to be evaluated on the submanifold $\{u_1=-\epsilon, u_2=1-\epsilon, u_3=s-\epsilon\}$.
\item Build the following 4x4 matrix:
\begin{align}\label{4dtrans}
\hat{\Psi}(v_1, v_2, v_3, v_4)= &\left(\begin{array}{cccc}
\phi_{11}(s)		&\phi_{12}(s)		&\phi_{12}(s)		&\phi_{13}(s)\\
\phi_{21}(s)		&\phi_{22}(s)		&\phi_{22}(s)		&\phi_{23}(s)\\
\phi_{31}(s)		&\phi_{32}(s)		&\phi_{32}(s)		&\phi_{33}(s)\\
0			&i 				&-i 				&0
\end{array}\right)\cdot\\\nonumber
					&\left(\begin{array}{cccc}
-\frac{(v_2-v_1)^{\mu}}{(\epsilon-s)^\mu} 	&0							&0											&0\\
0							&\frac{(v_2-v_1)^{\mu}}{(\epsilon-s)^\mu} 	&0											&0\\
0							&0							&\frac{1}{2}\frac{(v_2-v_1)^{-\mu}}{[\epsilon(\epsilon-1)]^\mu}		&0\\
0							&0							&0											&\frac{(v_2-v_1)^{-\mu}}{[\epsilon(\epsilon-1)]^\mu}
\end{array}\right)\cdot\\\nonumber
					&\left(\begin{array}{cccc}
\chi_{11}(\epsilon, s)		&\chi_{12}(\epsilon, s)			&0				&0\\
\chi_{21}(\epsilon, s)		&\chi_{22}(\epsilon, s)			&0				&0\\
0					&0						&\chi_{11}(\epsilon, s)	&\chi_{12}(\epsilon, s)\\
0					&0						&\chi_{21}(\epsilon, s)	&\chi_{22}(\epsilon, s)
\end{array}\right)
\end{align}
where $\phi_{i\alpha}$  are the elements of $\Phi$, and the variables $s, \epsilon$ are now to be re-interpreted as follows:
\begin{equation}\label{s, eps}
s=\frac{(v_3-v_1)(v_4-v_2)}{(v_2-v_1)(v_4-v_3)}\qquad\epsilon=\frac{v_2-v_4}{v_2-v_1}
\end{equation}
\end{itemize}
\begin{Thm} The matrix $\hat{\Psi}$ is the transition matrix of a 4-dimensional Frobenius manifold with tri-hamiltonian structure, written in the canonical coordinates $v_1, v_2, v_3, v_4$. Moreover, all such manifolds are obtained in this way.
\end{Thm}
Recall that the prepotential can be computed by quadrature from the transition matrix; it is uniquely determined up to Legendre-type transformations (\cite{Dub2Dtft}, Appendix B), which reflect the freedom in the choice of a selected column of $\hat{\Psi}$.\\
The distinctive feature of the construction is the appearance of the twisted periods of $M$ (\cite{DubAlmost*}) for the special value $\mu+1/2$ of the twisting parameter; they are related to the solutions of a certain two-component Fuchsian system, to which the full linear system for $\hat{\Psi}$ is naturally reduced.\\
The work is organized as follows: in Section 1 we introduce tri-hamiltonian structures, prove Lemma 1 and review some foundational material on the Darboux-Egoroff system and its interpretation in terms of isomonodromic deformations of Fuchsian operators. Section 2 is the core of the paper and contains the proofs of Theorems 1 and 2. The last section is devoted to examples: we apply the main construction to semisimple Frobenius structures arising in Hurwitz theory (\cite{DubMod}). There are three Hurwitz spaces of dimension 3; for two of them we propose (and prove explicitly in one case) a simple geometric interpretation for the $n=3/ n=4$  correspondence. The third (the $A_3$ singularity) looks quite different and in some sense more interesting. We compute it explicitly and provide some insight on its possible interpretation.\\
\\

I would like to thank my advisor Boris Dubrovin for suggesting to me the topic of this work and for his expert guidance throughout its preparation.
\section{Preliminaries}
\subsection{Tri-hamiltonan structures}
Let $M$ be a Frobenius manifold with prepotential $F(t^1, \dots, t^n)$. We will assume the standard normalization
\begin{equation}\label{antidiag}
\eta_{\alpha\beta}=\delta_{\alpha+\beta, n+1}
\end{equation}
of the metric, which can always be achieved by a suitable choice of flat coordinates provided $\eta(\partial/\partial t^1, \partial/\partial t^1)=0$. Introduce the \emph{grading operator}
\begin{equation}\label{hatmu}
\hat{\mu}\doteq \frac{2-d}{2}-\nabla\partial_E=\text{diag}(\mu_1, \dots, \mu_n)\qquad \mu_\alpha=\frac{2-d}{2}-d_\alpha
\end{equation}
where $d_\alpha$ are the homogeneity degrees of the flat coordinates, given by \eqref{Euler}. It follows from \eqref{antidiag} that
\begin{equation}\label{Nspec}
\mu_\alpha+\mu_{n+1-\alpha}=0
\end{equation}
Our starting point is the following result, relating the curvature of the metric \eqref{tildeeta} to the form of the grading operator:
\begin{Prop} The metric
$$
\tilde{\eta}^{\alpha\beta}=\eta^{\nu\alpha}(\mathcal{U}^2)^\beta_\nu\qquad \mathcal{U}=\partial_E\cdot
$$
is flat if and only if
\begin{equation}\label{WWDVV}
(\hat{\mu}^2(X\cdot Y))\cdot Z= X\cdot(\hat{\mu}^2(Y\cdot Z))
\end{equation}
for all vector fields $X, Y, Z$ on $M$.
\end{Prop}
\begin{Lemma} In the coordinates $t^\alpha,$ the contravariant Christoffel symbols of $\tilde{\eta}$ read
\begin{equation}\label{Christoffels}
\tilde{\Gamma}^{\alpha\beta}_\gamma=\sum_{\nu=1}^n(1-\mu_\beta+\mu_\nu)g^{\beta\nu}c^\alpha_{\nu\gamma}
\end{equation}
where $g$ is the intersection form \eqref{intform}.
\end{Lemma}
\begin{proof}
Recall that the contravariant Christoffel symbols are uniquely characterized by the equations
\begin{equation}\label{contravariantGamma}
\partial_\gamma\tilde{\eta}^{\alpha\beta}=\tilde{\Gamma}_\gamma^{\alpha\beta}+\tilde{\Gamma}_\gamma^{\beta\alpha}\qquad \tilde{\eta}^{\alpha\nu}\tilde{\Gamma}^{\beta\gamma}_\nu=\tilde{\eta}^{\beta\nu}\tilde{\Gamma}^{\alpha\gamma}_\nu
\end{equation}
The first is immediately checked using the identity
\begin{equation}\label{bubu}
\partial_\gamma\mathcal{U}_\beta^\alpha=(1-\mu_\alpha+\mu_\beta)c_{\beta\gamma}^\alpha
\end{equation}
which is a consequence of the quasi-homogeneity axiom and the normalization \eqref{Nspec}. The second follows from associativity since
\begin{align*}
\tilde{\eta}^{\alpha\nu}\tilde{\Gamma}_\nu^{\beta\gamma}&= \sum_{\lambda=1}^n(1-\mu_\gamma+\mu_\lambda)g^{\gamma\lambda} g^{\alpha\rho}(\mathcal{U}^\nu_\rho c^\beta_{\lambda\nu})\\
&=\sum_{\lambda=1}^n(1-\mu_\gamma+\mu_\lambda)g^{\gamma\lambda}\mathcal{U}^\nu_\lambda(g^{\alpha\rho} c^\beta_{\rho\nu})
\end{align*}
and the last expression is symmetric in $\alpha, \beta$ again by associativity.
\end{proof}
\begin{proof}[Proof of the Proposition] In terms of the contravariant Christoffel symbols, the Riemann curvature tensor reads
$$
R^{\alpha\beta\gamma}_\delta=\tilde{\Gamma}^{\alpha\beta}_\lambda\tilde{\Gamma}^{\lambda\gamma}_\delta- \tilde{\Gamma}^{\alpha\gamma}_\lambda\tilde{\Gamma}^{\lambda\beta}_\delta+\tilde{\eta}^{\alpha\lambda}\left(\partial_\lambda
\tilde{\Gamma}_\delta^{\beta\gamma}-\partial_\delta\tilde{\Gamma}_\lambda^{\beta\gamma}\right)
$$
The first piece vanishes since
$$
\tilde{\Gamma}^{\alpha\beta}_\lambda\tilde{\Gamma}^{\lambda\gamma}_\delta=\sum_{\nu, \rho}(1-\mu_\beta+\mu_\nu)(1-\mu_\gamma+\mu_\rho)g^{\beta\nu}g^{\gamma\rho}(c^{\alpha}_{\nu\lambda}c^\lambda_{\rho\delta})
$$
is symmetric in $\gamma, \beta$ by associativity. For the second piece, using \eqref{bubu} we compute
$$
\partial_\lambda\tilde{\Gamma}_\delta^{\beta\gamma}=\sum_{\nu=1}^n(1-\mu_\gamma+\mu_\nu)[(1-\mu_\gamma-\mu_\nu)c^{\gamma\nu}_\lambda c^\beta_{\nu\delta}+g^{\gamma\nu}\partial_\lambda c_{\nu\delta}^\beta]
$$
In the above expression everything is symmetric under the exchange of $\lambda$ and $\delta$, except for the term
$$
-\sum_{\nu=1}^n\mu_\nu^2c^{\gamma\nu}_\lambda c^\beta_{\nu\delta}=-\tilde{\eta}^{\gamma\rho}[\hat{\mu}^2(\partial_\lambda\cdot\partial_\rho)\cdot\partial_\delta]^\beta
$$
The proposition is proved.
\end{proof}
Recall that two contravariant metrics $g_1, g_2$ form a \emph{flat pencil} if
\begin{itemize}
\item[(i)] For generic $\epsilon$ the metric
$$
g_{\epsilon}=g_1+\epsilon g_2
$$
is flat.
\item[(ii)] The contravariant Christoffel symbols of $g_{\epsilon}$ have the form
$$
\Gamma_{\epsilon}=\Gamma_1+\epsilon\Gamma_2
$$
where $\Gamma_i$ are the Christoffel symbols of $g_i$.
\end{itemize}
\begin{Cor} Under the condition \eqref{WWDVV}, $\tilde{\eta}$, $g$ and $\eta$ form a 2-parameters flat pencil of metrics.
\end{Cor}
\begin{proof}
The identity
$$
\frac{\partial}{\partial t^1}\mathcal{U}=I
$$
yields
$$
\frac{\partial}{\partial t^1}\tilde{\eta}^{\alpha\beta}=2g^{\alpha\beta}, \qquad \frac{\partial^2}{(\partial t^1)^2}\tilde{\eta}^{\alpha\beta}=2\eta^{\alpha\beta}
$$
Therefore the 2-parameters pencil
$$
g_{\epsilon_1, \epsilon_2}=\tilde{\eta}+\epsilon_1 g+\epsilon_2 \eta
$$
can be expressed as
$$
g_{\epsilon_1, \epsilon_2}^{\alpha\beta}(t)=\tilde{\eta}^{\alpha\beta}(t^1+\frac{\epsilon_1}{2}, t^2, \dots, t^n)+\left(\epsilon_2-\frac{\epsilon_1^2}{4}\right)\eta^{\alpha\beta}
$$
Then one checks directly, substituting in \eqref{contravariantGamma}, that the contravariant Christoffel symbols of $g_{\epsilon_1, \epsilon_2}$ are given by
$$
\Gamma_{\epsilon_1, \epsilon_2}(t)=\tilde{\Gamma}(t^1+\frac{\epsilon_1}{2}, t^2, \dots, t^n)=\tilde{\Gamma}(t)+\epsilon_1\Gamma(t)
$$
where $\Gamma$ are the Christoffel symbols of the intersection form, 
$$
\Gamma^{\alpha\beta}_\gamma=\left(\frac{1}{2}-\mu_\beta\right)c^{\alpha\beta}_\gamma
$$
(see \cite{Dub2Dtft}). At this point flatness of the whole pencil follows immediately from the flatness of $\tilde{\eta}$.
\end{proof}
It is clear that condition (C) in the Introduction implies \eqref{WWDVV}: it means precisely that the square of the grading operator is scalar, 
$$
\hat{\mu}^2=\mu^2 I
$$
We use the latter slightly stronger notion as the definition of tri-hamiltonian Frobenius manifold because it is more manageable, and coincides with \eqref{WWDVV} under the assumption of semisimplicity, which we will add henceforth. It also coincides with Pavlov-Tsarev's original definition in \cite{PT}.
\subsection{Semisimplicity and Darboux-Egoroff} A Frobenius manifold $M$ is called \emph{semisimple} (or \emph{massive}) if at any point of $M$ the Frobenius algebra on the tangent space contains no nilpotents.  In this case one locally constructs \emph{canonical coordinates}  $u_1, \dots, u_n$ reducing the multiplication table to the standard semisimple form
\begin{equation}\label{prod}
\frac{\partial}{\partial u_i}\cdot\frac{\partial}{\partial u_j}=\delta_{ij}\frac{\partial}{\partial u_i}
\end{equation}
and the metric to diagonal form
\begin{equation}\label{diageta}
\eta =\sum_{i=1}^n \eta_i(u) du_i^2
\end{equation}
whereas the unit and the Euler vector field become
\begin{equation}\label{eE}
\partial_e=\sum_{i=1}^n\frac{\partial}{\partial u_i}\qquad\partial_E=\sum_{i=1}^nu_i\frac{\partial}{\partial u_i}
\end{equation}
respectively. The flat and canonical coordinate systems provide complementary descriptions of the Frobenius structure: in the first the metric is constant and the geometric data of the manifold consist of the product, which contains the prepotential. In the second the product is reduced to normal form, and the non-trivial object becomes the metric; in this new picture the role of WDVV is replaced by a classical differential-geometric system. The relation between the two pictures is described by the \emph{transition matrix} \eqref{transitionPsi}.\\
\indent More precisely: to the diagonal metric \eqref{diageta}, we associate the off-diagonal matrix $\Gamma=(\gamma_{ij})$ of \emph{rotation coefficients}, defined by
\begin{equation}\label{gamma}
\gamma_{ij}(u)\doteq\frac{1}{\sqrt{\eta_i(u)}}\frac{\partial }{\partial u_i}\sqrt{\eta_j(u)}\qquad i\neq j
\end{equation}
Introduce the matrix
\begin{equation}\label{defV}
V(u)\doteq[\Gamma(u), U]\qquad U\doteq\text{diag}(u_1, \dots, u_n)
\end{equation}
\begin{Thm}[Dubrovin] The matrix $V$ coincides with the grading operator written in the frame of normalized idempotents:
\begin{equation}\label{Vmuhat}
V=\Psi\hat{\mu}\Psi^{-1}
\end{equation}
It satisfies the system of nonlinear equations
\begin{equation}\label{DE}
\frac{\partial}{\partial u_i}V=[V_i, V]\qquad  i=1, \dots, n
\end{equation}
$$
V_i\doteq [\Gamma, E_i]\equiv \text{ad}_{E_i}\text{ad}_U^{-1}V
$$
where $E_i=(\delta_{ij}\delta_{ik})$. In particular, the transition matrix satisfies the linear system
\begin{equation}\label{fundlinsys}
\frac{\partial}{\partial u_i}\Psi=V_i\Psi\qquad i=1,\dots, n
\end{equation}
\end{Thm}
Formula \eqref{DE} is known as the \emph{Darboux-Egoroff system}.\\
Conversely, let $V$ be any skew-symmetric $n\times n$ solution of \eqref{DE}. Then the linear system
\begin{equation}\label{linsys}
\frac{\partial}{\partial u_i}\psi=V_i\psi\qquad i=1, \dots, n
\end{equation}
is compatible, i.e. it possesses an $n$-dimensional space of solutions. Moreover, \eqref{linsys} preserves the eigenvectors and the eigenspaces of $V$: assuming $V$ to be diagonalizable, there exists a basis of solutions $\psi_{(\alpha)}=(\psi_{1\alpha}, \dots, \psi_{n\alpha})^t$, $\alpha=1, \dots, n$ satisfying
$$
V\psi_{(\alpha)}=\partial_E\psi_{(\alpha)}=\mu_{\alpha}\psi_{(\alpha)}
$$
Here $\partial_E$ is the Euler vector field in the canonical coordinates and the constants $\mu_\alpha$ are the eigenvalues of $V$. Therefore the solutions $\psi_{(\alpha)}$ are homogeneous functions of degree the corresponding eigenvalue. In this way we obtain an \emph{homogeneous} fundamental solution $\Psi=(\psi_{i\alpha})$ of \eqref{linsys}. We can assume without losing generality that none of the component $\psi_{i\alpha}$ vanishes identically, otherwise the system \eqref{DE} would split.\\
Select now an element of the basis, label it as $\psi_{(1)}$, and (after restricting to the open set $\psi_{i1}\neq0, i=1, \dots, n$) substitute in the formulas
\begin{equation}\label{reconstruction}
\eta_{\alpha\beta}\doteq\sum_{i=1}^n\psi_{i\alpha}\psi_{i\beta}
\end{equation}
$$
dt^\alpha\doteq\sum_{i=1}^n\eta^{\alpha\beta}\psi_{i1}\psi_{i\beta}du_i
$$
$$
c_{\alpha\beta\gamma}\doteq\sum_{i=1}^n\frac{\psi_{i\alpha}\psi_{i\beta}\psi_{i\gamma}}{\psi_{i1}}
$$
Then $\eta_{\alpha\beta}$ is a constant and nondegenerate matrix, and $c_{\alpha\beta\gamma}$ are third derivatives with respect to the $t^\alpha$ of a prepotential $F$.\\
The above discussion is summarized in the following standard result:
\begin{Thm}[Dubrovin] There is a 1:1 correspondence between semisimple Frobenius manifolds and diagonalizable solutions of the Darboux-Egoroff system with a marked eigenvector.
\end{Thm}
It is interesting that the constraint (A), (B) of the Introduction assume a straightforward meaning once expressed in terms of the Darboux-Egoroff system: indeed \eqref{DE} immediately implies
\begin{equation}\label{inv1}
\partial_e V=0, \qquad \partial_E V=0
\end{equation}
so that any solution is invariant under affine transformations $u_i\to au_i+b$.\\
It turns out condition (C) means precisely invariance of $V$ under the full group of fractional linear transformations
\begin{equation}\label{flt}
u_i\to\frac{au_i+b}{cu_i+d}\qquad ad-bc\neq0
\end{equation}
Indeed, the vector fields $\partial_e, \partial_E$ can be completed to the $sl(2)$ algebra generating \eqref{flt} by adding the vector field
$$
\partial_{\tilde{e}}\doteq\sum_{i=1}^nu_i^2\frac{\partial}{\partial u_i}\equiv \partial_E\cdot\partial_E
$$
\begin{Lemma} For a solution $V$ of \eqref{DE}, the following conditions are equivalent
\begin{itemize}
\item[(i)] The matrix $V^2$ is diagonal.
\item[(ii)] The matrix $V$ satisfies
\begin{equation}\label{inv2}
\partial_{\tilde{e}}V=0 
\end{equation}
\end{itemize}
\end{Lemma}
\begin{proof}
This is a simple calculation:
$$
\partial_{\tilde{e}} V=\sum_{i=1}^nu_i^2[V_i, V]=[[\Gamma, U^2], V]=[UV, V]+[VU, V]=[U, V^2]
$$
\end{proof}
Assuming that for the solution $V$ the system \eqref{DE} does not split (that is, $V$ is not block-diagonal), condition (i) is the same as
\begin{equation}\label{Vsquared}
V^2=\mu^2 I
\end{equation}
Equivalently, $V$ must be an skew-symmetric diagonalizable matrix with only 2 eigenvalues $\pm\mu$, which is exactly our definition of tri-hamiltonianity (recall that  by \eqref{Vmuhat} the eigenvalues of $V$ coincide with those of the grading operator). It is clear from the Lax form of Darboux-Egoroff that this last condition is preserved by the system, i.e. it indeed defines a reduction.
\begin{Rem} It is worth noting that equation \eqref{inv2} implies flatness of the third metric even without any homogeneity assumption. In other words, if the rotation coefficients of a diagonal Egoroff metric $\eta$ satisfy
$$
\partial _k\gamma_{ij}=\gamma_{ik}\gamma_{kj}\qquad i, j, k\;\; \text{distinct}
$$
$$
\partial_e\gamma_{ij}=0
$$
and the additional constraint
$$
\partial_{\tilde{e}}\gamma_{ij}=-(u_i+u_j)\gamma_{ij}
$$
(but not necessarily the homogeneity $\partial_E\gamma_{ij}=-\gamma_{ij}$), then diagonal metric
$$
\tilde{\eta}=\sum_{i=1}^n\frac{\eta_i}{u_i^2}du_i^2
$$
is flat (and compatible with $\eta$).
\end{Rem}
In the lowest dimensional case $n=2$, every Frobenius manifold has tri-hamiltonian structure, since \eqref{inv1} completely determines the solution
$$
V=\mu\left(\begin{array}{cc}
0	&1\\
-1	&0
\end{array}\right)
$$
up to the constant $\mu$, and \eqref{inv2} is automatically satisfied. In the next section we study the first nontrivial case $n=4$ in details\\
\indent Before proceeding we will review a few more facts from the theory of Frobenius manifolds, specifically concerning the connection between Darboux-Egoroff and isomonodromic deformations of Fuchsian linear operators. First we introduce, in addition to \eqref{prod}, a second product between vector fields by the rule\footnote{Let us stress that whenever needed we assume implicitly to be working ``outside of the discriminant", i.e. in a domain where the canonical coordinates do not vanish. This condition is also necessary for non-degeneracy of the metrics $\tilde{\eta}, g$.}
\begin{equation}\label{prod*}
\frac{\partial}{\partial u_i}\star\frac{\partial}{\partial u_j}=\frac{1}{u_i}\delta_{ij}\frac{\partial}{\partial u_i}
\end{equation}
and a pencil of affine connections depending on a complex parameter $\nu$
$$
\tilde{\nabla}^{\nu}_XY\doteq\tilde{\nabla}_XY+\nu X\star Y
$$
where $X, Y$ are vector fields and $\tilde{\nabla}$ is the Levi-Civita connection of the intersection form. This is called the \emph{almost-dual deformed connection} and is proved in \cite{DubAlmost*} to be identically flat in $\nu$. In particular there exist an $n$-dimensional space of flat coordinates
$$
p=p(t^1, \dots, t^n;\nu)\;\;\; \text{such that}\;\;\; \tilde{\nabla}^{\nu}dp=0
$$
(expressed here as function of the flat coordinates $t^\alpha$ of $\eta$), called the \emph{twisted periods} of the Frobenius manifold. The parameter $\nu$ is called twisting parameter.\\
The next result from \cite{DubAlmost*} plays a central role in the following:
\begin{Thm}[Dubrovin] Let $p(t;\nu)$ be a twisted period and
$$
\chi=(\chi^1, \dots, \chi^n)^t, \qquad\chi^\alpha\doteq\eta^{\alpha\beta}\frac{\partial}{\partial t^\beta}p
$$
be the gradient of $p$ with respect to the coordinates $t^\alpha$. Introduce dependence on a shift parameter $\epsilon$ in $\chi$ by setting
$$
\chi(t;\epsilon)\doteq\chi(t^1-\epsilon, t^2, \dots, t^n)
$$
Then $\chi$ satisfies the linear system
\begin{equation}\label{FuchSys}
\frac{\partial}{\partial\epsilon}\chi=\sum_{i=1}^n\frac{R_i}{\epsilon-u_i}\chi
\end{equation}
\begin{equation}
\frac{\partial}{\partial u_i}\chi=-\frac{R_i}{\epsilon-u_i}\chi, \qquad i=1, \dots, n
\end{equation}
with residue matrices
\begin{equation}\label{Ri}
R_i=-\Psi^{-1}E_i\Psi\left(\frac{1}{2}-\nu+\hat{\mu}\right)
\end{equation}
In particular, the $(u_1, \dots, u_n)$-parametric family of Fuchsian differential operators
\begin{equation}\label{FuchOp}
\frac{d}{d\epsilon}-\sum_{i=1}^n\frac{R_i}{\epsilon-u_i}
\end{equation}
is isomonodromic.
\end{Thm}
\begin{Rem}The converse statement also holds. Indeed, the Schlesinger equations for the family \eqref{FuchOp} coincide with the Darboux-Egoroff system and do not depend on $\nu$ (although the solutions $\chi$ of the Fuchsian system of course do).
\end{Rem}
Let us consider the special case $n=3, \nu=1/2$. We have
\begin{equation}\label{t2decoupling}
\hat{\mu}=\mu\left(\begin{array}{ccc}
1	&0	&0\\
0	&0	&0\\
0	&0	&-1
\end{array}\right), \qquad R_i=\mu\left(\begin{array}{ccc}
-\psi_{i1}\psi_{i3}	&0	&\psi_{i3}^2\\
-\psi_{i1}\psi_{i2}	&0	&\psi_{i2}\psi_{i3}\\
-\psi_{i1}^2		&0	&\psi_{i1}\psi_{i3}
\end{array}\right)
\end{equation}
Thus the Fuchsian system \eqref{FuchSys} splits into a quadrature and the 2x2 system
\begin{equation}\label{2x21}
\frac{\partial}{\partial\epsilon}\chi=\sum_{i=1}^3\frac{A_i}{\epsilon-u_i}\chi, \qquad A_i=\mu\left(
\begin{array}{cc}
-\psi_{i1}\psi_{i3}	&\psi_{i3}^2\\
-\psi_{i1}^2		&\psi_{i1}\psi_{i3}
\end{array}\right)
\end{equation}
This is a two-components Fuchsian system with 4 simple poles (at  the canonical coordinates and $\infty$), and by a classical result (\cite{Sch}) the isomonodromy equations for such system are equivalent to the Painlev\'e VI equation (\cite{Painl}, \cite{Gamb}).\\ Explicitly, introduce the invariant parameter 
$$
s\doteq\frac{u_3-u_1}{u_2-u_1}
$$
and the rescaled transition matrix
\begin{equation}\label{phi}
\Phi(s)=\Psi(u_1, u_2, u_3)\;(u_2-u_1)^{-\hat{\mu}}\equiv \Psi(0, 1, s)
\end{equation}
Then after the gauge transformation
$$
\chi\to\left(\begin{array}{cc}
(u_2-u_1)^{-\mu}	&0\\
0				&(u_2-u_1)^\mu
\end{array}\right)\chi
$$
and the change of variable
$$
\epsilon\to(u_2-u_1)\epsilon+u_1
$$
the system \eqref{2x21} takes the standard form
\begin{equation}\label{2x22}
\frac{\partial}{\partial\epsilon}\chi=\left(\frac{A_1}{\epsilon}+\frac{A_2}{\epsilon-1}+\frac{A_3}{\epsilon-s}\right)\chi, \qquad A_i(s)=\mu\left(
\begin{array}{cc}
-\phi_{i1}\phi_{i3}	&\phi_{i3}^2\\
-\phi_{i1}^2		&\phi_{i1}\phi_{i3}
\end{array}\right)
\end{equation}
(we re-denoted here by the same letter the new residue matrices and independent variable). The residue at infinity is
$$
A_4=-(A_1+A_2+A_3)=\mu\left(\begin{array}{cc}
1	&0\\
0	&-1
\end{array}\right)
$$
Since the top-right entry of $A_4$ is zero, the top right entry of
$$
\epsilon(\epsilon-1)(\epsilon-s)
\left(\frac{A_1}{\epsilon}+\frac{A_2}{\epsilon-1}+\frac{A_3}{\epsilon-s}\right)
$$
is a degree one polynomial in $\epsilon$, with $s$-dependent coefficients; let $y(s)$ denote the unique zero of this polynomial. Then the system \eqref{2x22} is isomonodromic if and only if $y$ satisfies the special version \eqref{PVImu} of Painlev\'e VI, depending on a single parameter $\mu$.
\begin{Cor} The Darboux-Egoroff system in dimension $n=3$ is equivalent to the Painlev\'e VI$\mu$ equation.
\end{Cor}
An explicit parametric representation for the prepotential of the semisimple 3-dimensional Frobenius manifold in terms of the corresponding Painlev\'e VI$\mu$ transcendent was found by Guzzetti (\cite{Guzz}). An alternative approach, based on a representation of PVI in terms of elliptic functions, was proposed by Manin (\cite{Man}).
\section{Tri-hamiltonian structures in dimension 4}
We now turn to the study of 4-dimensional Frobenius manifold with tri-hamil\-tonian structure. A simple observation shows that again they correspond to solutions of PVI$\mu$:
\begin{Prop}Let
\begin{equation}\label{3d}
V(u_1, u_2, u_3)=\left(\begin{array}{ccc}
0	&-c(s)	&b(s)\\
c(s)	&0		&-a(s)\\
-b(s)	&a(s)	&0
\end{array}\right)\qquad s=\frac{u_3-u_1}{u_2-u_1}
\end{equation}
be a 3-dimensional solution of the Darboux-Egoroff system. Then
\begin{equation}\label{W}
W(v_1, v_2, v_3, v_4)=\left(\begin{array}{cccc}
0	&-c(s)	&b(s)	&-a(s)\\
c(s)	&0		&-a(s)	&-b(s)\\
-b(s)	&a(s)	&0		&-c(s)\\
a(s)	&b(s)	&c(s)		&0
\end{array}\right)\qquad s=\frac{(v_3-v_1)(v_4-v_2)}{(v_2-v_1)(v_4-v_3)}
\end{equation}
is a 4-dimensional solution with tri-hamiltonian structure of the same system. Moreover, any such 4-dimensional solution is of the form \eqref{W}, up to changing the sign of the last row and column.
\end{Prop}
\begin{proof}
The Darboux-Egoroff system for \eqref{3d} reads
\begin{equation}\label{abcODE}
\frac{d}{ds}a=\frac{bc}{s}\;\qquad\; \frac{d}{ds}b=\frac{ac}{1-s}\;\qquad\; \frac{d}{ds}c=\frac{ab}{s(s-1)}
\end{equation}
Plugging \eqref{W} into \eqref{DE} we again obtain \eqref{abcODE}. The invariance under fractional linear transformations of  $v_i$ is clear from the fact the $W$ only depends on the cross-ratio $s$.\\
The converse statement follows from the observation  that any skew-symmetric 4x4 matrix squaring to a multiple of the identity (see \eqref{Vsquared}) is of the form \eqref{W}, up to the overall sign of the last row and column. Indeed, \eqref{tri-system2} implies that $W/(i\mu)$ is orthogonal, and the only skew-symmetric matrices in $SO(4)$ are off-diagonal left and right isoclinic rotations. 
\end{proof}
\begin{Rem} Throughout the section, we fix the solution \eqref{3d} and work simultaneously with the 3-dimensional Frobenius manifold associated to it and with the 4-dimensional manifold associated to \eqref{W}. As above, we denote by $u_i, v_i$ the respective canonical coordinate, but use the same letter $s$ for the respective invariant parameters. Here the ``invariance" is meant under affine transformations of the $u_i$, and under fractional linear transformations of the $v_i$.
\end{Rem}
As a direct corollary we obtain Theorem 1; it suffices to note that \eqref{tri-system} coincides with \eqref{FuchOp} (for $n=4$) after gauging by the transition matrix. Thus we established a correspondence between 3-dimensional semisimple Frobenius manifolds and 4-dimensional semisimple Frobenius manifolds with tri-hamiltonian structure. Proposition 2 shows that the relation between the respective solutions of Darboux-Egoroff is extremely simple. Our next object of study is the 4-dimensional prepotential associated to \eqref{W}; it turns out that this second object is related to the 3-dimensional manifold in a highly nontrivial way.\\
The main goal is to describe the homogeneous fundamental solution of the linear system
\begin{equation}\label{Wlinsys}
\frac{\partial}{\partial v_i}\psi= W_i\psi\qquad W_i=\text{ad}_{E_i}\text{ad}_{\hat{U}}^{-1}W
\end{equation}
(here $\hat{U}=\text{diag}(v_1, \dots v_4)$) i.e. the transition matrix of the 4-dimensional manifold. From the latter we can derive (up to quadrature) the tri-hamiltonian prepotential using the formulas \eqref{reconstruction}. The idea is to write \eqref{Wlinsys} in a basis of eigenvectors of $W$ built out of 3-dimensional data and depending only on the cross-ratio, namely:
\begin{equation}\label{hatphi}
\hat{\phi}_{(1)}(s)=\left(\begin{array}{c}
\phi_{11}\\
\phi_{21}\\
\phi_{31}\\
0
\end{array}\right)\qquad
\hat{\phi}_{(2)}(s)=\left(\begin{array}{c}
\phi_{12}\\
\phi_{22}\\
\phi_{32}\\
i
\end{array}\right)
\end{equation}
$$
\hat{\phi}_{(3)}(s)=\left(\begin{array}{c}
\phi_{12}\\
\phi_{22}\\
\phi_{23}\\
-i
\end{array}\right)\qquad
\hat{\phi}_{(4)}(s)=\left(\begin{array}{c}
\phi_{13}\\
\phi_{23}\\
\phi_{33}\\
0
\end{array}\right)
$$
where where $\phi_{i\alpha}$ have been defined in $\eqref{phi}$. It is readily checked that\footnote{Note that choosing the opposite sign in the last row and column of $W$ simply interchanges the eigenvectors $\hat{\phi}_{(2)}$ and $\hat{\phi}_{(3)}$. At the end of the construction, this results in changing the sign of the last column of \eqref{4dtrans}. For simplicity we will disregard this ambiguity in the following.}
$$
W\hat{\phi}_{(\alpha)}=\mu\hat{\phi}_{(\alpha)}, \qquad\alpha=1, 2
$$
$$
W\hat{\phi}_{(\alpha)}=-\mu\hat{\phi}_{(\alpha)}, \qquad\alpha=3, 4
$$
where $\mu=-(a^2+b^2+c^2)$. Recall that each eigenspace is invariant under the linear system, and that solutions lying inside an eigenspace are homogeneous functions of degree the corresponding eigenvalue.\\
Let us restrict our attention to the $\mu$-eigenspace: the linear system implies
$$
\partial_e\psi=0\qquad\partial_E\psi=W\psi=\mu\psi
$$
Thus, if we rescale $\psi$ by $(v_2-v_1)^{-\mu}$ the result is invariant under affine transformations, and we can express it as a function of the cross-ratio $s$ and an additional affine-invariant independent parameter. We chose it to be
\begin{equation}\label{epsilon}
\epsilon\doteq\frac{v_2-v_4}{v_2-v_1}
\end{equation}
Summarizing, we have
\begin{equation}\label{psi1}
\psi=(v_2-v_1)^\mu[\alpha_1(\epsilon, s)\hat{\phi}_{(1)}(s)+\alpha_2(\epsilon, s)\hat{\phi}_2(s)]
\end{equation}
Similarly, in the $-\mu$-eigenspace
\begin{equation}\label{psi2}
\psi=(v_2-v_1)^{-\mu}[\beta_1(\epsilon, s)\hat{\phi}_{(3)}(s)+\beta_2(\epsilon, s)\hat{\phi}_4(s)]
\end{equation}
We are now ready to prove our main result:
\begin{Thm} The homogeneous solutions of the linear system \eqref{Wlinsys} are in 1:1 correspondence with the solutions of the Fuchsian system
\begin{equation}\label{2x2new}
\frac{\partial}{\partial\epsilon}\chi=\left(\frac{B_1}{\epsilon}+\frac{B_2}{\epsilon-1}+\frac{B_3}{\epsilon-s}\right)\chi
\end{equation}
$$
B_i(s)=\mu\left(\begin{array}{cc}
\phi_{i2}^2		&2\phi_{i2}\phi_{i3}\\
\phi_{i1}\phi_{i2}	&2\phi_{i1}\phi_{i3}
\end{array}\right)
$$
\end{Thm}
\begin{proof}
We begin by observing that the linear system implies
\begin{equation}\label{tildeW}
\partial_{\tilde{e}}\psi=\tilde{W}\psi, \qquad \tilde{W}=UW+WU
\end{equation}
Since $W^2=\mu^2I$, $\tilde{W}$ commutes with $W$ and therefore is block diagonal in the basis $\hat{\phi}_{(\alpha)}$ \eqref{hatphi}. Let us compute it explicitly. If $\hat{\Phi}=(\hat{\phi}_{i\alpha})$, it follows from the normalization
\begin{equation}\label{norm}
\sum_{i=1}^3\phi_{i\alpha}\phi_{i\beta}=\eta_{\alpha\beta}=\delta_{\alpha+\beta, 4}
\end{equation}
that
$$
\hat{\Phi}^{-1}=\left(\begin{array}{cccc}
\phi_{13}		&\phi_{23}	&\phi_{33}	&0\\
\phi_{12}/2	&\phi_{22}/2	&\phi_{32}/2	&-i/2\\
\phi_{12}/2	&\phi_{22}/2	&\phi_{32}/2	&i/2\\
\phi_{11}		&\phi_{21}	&\phi_{31}	&0\\
\end{array}\right)
$$
Then a simple computation yields
$$
\hat{\Phi}^{-1}\tilde{W}\hat{\Phi}=\mu\left(\begin{array}{cccc}
2G_{13}	&2G_{23}		&0			&0\\
G_{12}	&G_{22}+v_4	&0			&0\\
0		&0			&-G_{22}-v_4	&-G_{23}\\
0		&0			&-2G_{12}	&-2G_{13}
\end{array}\right)
$$
where
$$
G_{\alpha\beta}\doteq\sum_{i=1}^3v_i\phi_{i\alpha}(s)\phi_{i\beta}(s)
$$
Let us restrict to the first eigenspace. Here and in the following we denote for brevity $v_{ij}\equiv v_i-v_j$. Substituting \eqref{psi1} into \eqref{tildeW} we get
\begin{align*}
\partial_{\tilde{e}}\psi		&=\left(\mu v_{21}^{\mu-1}\partial_{\tilde{e}}v_{21}+v_{21}\partial_{\tilde{e}}\epsilon\frac{\partial}{\partial\epsilon}\right)\left(\begin{array}{c}\alpha_1\\ \alpha_2\end{array}\right)=\\
					&=v_{21}^\mu\mu\left(\begin{array}{cc}
2G_{13}	&2G_{23}\\
G_{12}	&G_{22}+v_4
\end{array}\right)\left(\begin{array}{c}\alpha_1\\ \alpha_2\end{array}\right)
\end{align*}
since $\partial_{\tilde{e}}s=0$. Using $\partial_{\tilde{e}}v_{21}=v_{21}(v_1+v_2)$ and $\partial_{\tilde{e}}\epsilon=-v_{21}\epsilon(\epsilon-1)$ (see \eqref{epsilon}), we arrive at the formula
\begin{equation}\label{C}
\frac{\partial}{\partial\epsilon}\alpha=C(\epsilon, s)\alpha\qquad \alpha\equiv\left(\begin{array}{c}\alpha_1\\ \alpha_2\end{array}\right)
\end{equation}
$$
C(\epsilon, s)=\frac{\mu}{\epsilon(\epsilon-1)v_{21}}\left(\begin{array}{cc}
v_1+v_2-2G_{13}	&-2G_{23}\\
-G_{12}			&v_1+v_2-G_{22}-v_4
\end{array}\right)
$$
Equation \eqref{C} takes a nice form once we explicit $C$ as a function of $s$ and $\epsilon$. For example, for the top left entry
\begin{align*}
C_{11}	&=\frac{\mu}{\epsilon(\epsilon-1)v_{21}}\Big(v_{21}\phi_{11}\phi_{13}+v_{12}\phi_{21}\phi_{23}+(v_{13}+v_{23})\phi_{31}\phi_{33}\Big)=\\
		&=\frac{\mu}{\epsilon(\epsilon-1)}\left[\phi_{11}\phi_{13}-\phi_{21}\phi_{23}+\left(-s\frac{\epsilon-1}{\epsilon-s}+\epsilon\frac{1-s}{\epsilon-s}\right)\phi_{31}\phi_{33}\right]=\\
		&=\mu\left[\frac{-\phi_{11}\phi_{13}+\phi_{21}\phi_{23}+\phi_{31}\phi_{33}}{\epsilon}+\frac{\phi_{11}\phi_{13}-\phi_{21}\phi_{23}+\phi_{31}\phi_{33}}{\epsilon-1}+\frac{-2\phi_{31}\phi_{33}}{\epsilon-s}\right]=\\
		&=\mu\left[\frac{1}{\epsilon}\phi_{12}^2+\frac{1}{\epsilon-1}\phi_{22}^2+\frac{1}{\epsilon-s}(\phi_{32}^2-1)\right]
\end{align*}
where we used the normalization \eqref{norm} in the last line, as well as the inverse relation $\eta^{\alpha\beta}\phi_{i\alpha}\phi_{j\beta}=\delta_{ij}$ (in particular, $2\phi_{i1}\phi_{i3}+\phi_{i2}^2=1$).
Proceeding in this way, we obtain
$$
C(s, \epsilon)=\frac{C_1(s)}{\epsilon}+\frac{C_2(s)}{\epsilon-1}+\frac{C_3(s)}{\epsilon-s}
$$
$$
C_1=\mu\left(\begin{array}{cc}
\phi_{12}^2		&-2\phi_{12}\phi_{13}\\
-\phi_{11}\phi_{12}	&2\phi_{11}\phi_{13}
\end{array}\right)
$$
$$
C_2=\mu\left(\begin{array}{cc}
\phi_{22}^2		&-2\phi_{22}\phi_{23}\\
-\phi_{21}\phi_{22}	&2\phi_{21}\phi_{23}
\end{array}\right)
$$
$$
C_3=\mu\left(\begin{array}{cc}
\phi_{32}^2-1		&-2\phi_{32}\phi_{33}\\
-\phi_{31}\phi_{32}	&2\phi_{31}\phi_{33}-1
\end{array}\right)
$$
Thus we see that the system \eqref{C} is the same as \eqref{2x2new}, up to the sign of the off-diagonal entries of the residue matrices and a shift by $-\mu I$ in $C_3$: the gauge transformation
$$
\alpha\to\chi=(\epsilon-s)^\mu\left(\begin{array}{cc}
-1	&0\\
0	&1
\end{array}\right)\alpha
$$
transforms one system into the other.\\
Repeating the same procedure in the second eigenspace (i.e. starting from \eqref{psi2}) we obtain
$$
\frac{\partial}{\partial\epsilon}\beta=\left(\frac{D_1}{\epsilon}+\frac{D_2}{\epsilon-1}+\frac{D_3}{\epsilon-s}\right)\beta\qquad \beta\equiv\left(\begin{array}{c}\beta_1\\ \beta_2\end{array}\right)
$$
$$
D_1=\mu\left(\begin{array}{cc}
\phi_{12}^2-1		&\phi_{12}\phi_{13}\\
2\phi_{11}\phi_{12}	&2\phi_{11}\phi_{13}-1
\end{array}\right)
$$
$$
D_2=\mu\left(\begin{array}{cc}
\phi_{22}^2-1		&\phi_{22}\phi_{23}\\
2\phi_{21}\phi_{22}	&2\phi_{21}\phi_{23}-1
\end{array}\right)
$$
$$
D_3=\mu\left(\begin{array}{cc}
\phi_{32}^2		&\phi_{32}\phi_{33}\\
2\phi_{31}\phi_{32}	&2\phi_{31}\phi_{33}
\end{array}\right)
$$
and the gauge transformation
$$
\beta\to\chi=\epsilon^\mu(\epsilon-1)^\mu\left(\begin{array}{cc}
2	&0\\
0	&1
\end{array}\right)\beta
$$
again leads to \eqref{2x2new}. This completes the proof.
\end{proof}
It can be checked directly that the system \eqref{2x2new} is isomonodromic in $s$ (this is anyway a direct consequence of the next proposition), i.e.
$$
\frac{\partial}{\partial s}\chi=-\frac{B_3}{\epsilon-s}\chi
$$
Consequently, it is described by some Painlev\'e VI transcendent; more precisely, since the matrices $B_i$ have eigenvalues $\mu, 0$ and the residue at infinity is
$$
B_4=-(B_1+B_2+B_3)=\mu\left(\begin{array}{cc}
-1	&0\\
0	&-2
\end{array}\right),
$$
the unique zero $y(s)$ of the top right entry of
$$
\epsilon(\epsilon-1)(\epsilon-s)\left(\frac{B_1}{\epsilon}+\frac{B_2}{\epsilon-1}+\frac{B_3}{\epsilon-s}\right)
$$
solves the Pailev\'e VI equation
\begin{align}\label{PVImu'}
\frac{d^2y}{ds^2}=	&\frac{1}{2}\left(\frac{1}{y}+\frac{1}{y-1}+\frac{1}{y-s}\right)\left(\frac{dy}{ds}\right)^2-\left(\frac{1}{s}+\frac{1}{s-1}+\frac{1}{y-s}\right)\frac{dy}{ds}+\\
				&+\frac{1}{2}\frac{y(y-1)(y-s)}{s^2(s-1)^2}\left((\mu-1)^2-\frac{\mu^2s}{y^2}+\frac{\mu^2(s-1)}{(y-1)^2}+\frac{(1-\mu^2)s(s-1)}{y-s}^2\right)\nonumber
\end{align}
which is obtained from PVI$\mu$ by a certain Okamoto transformation (\cite{Okam}). This suggests that the systems \eqref{2x22} and \eqref{2x2new} are closely related; the next result clarifies this relation, and also provides an explicit description of the solutions of \eqref{2x2new}.
\begin{Prop} The system \eqref{2x2new} is equivalent to \eqref{FuchSys} for $n=3, \nu=\mu+1/2$. More precisely, let $p(t, \mu+1/2)$ be a twisted period of the 3-dimensional Frobenius manifold with twisting parameter $\nu=\mu+1/2$, and
$$
\chi=\left(\begin{array}{c}\chi_1\\ \chi_2\end{array}\right)\doteq\left(\begin{array}{c}\partial p/\partial t^2\\ \partial p/\partial t^1\end{array}\right)
$$
Introduce dependence on a shift parameter $\epsilon$ in $\chi$ by setting
$$
\chi(t;\epsilon)\doteq\chi(t^1-\epsilon, t^2, t^3)
$$
Then, after the rescaling
$$
\chi\to\left(\begin{array}{cc}
(u_2-u_1)^{-\mu}	&0\\
0				&1
\end{array}\right)\chi
$$
and the change of variable
$$
\epsilon\to(u_2-u_1)\epsilon+u_1
$$
we obtain a solution of \eqref{2x2new}.
\end{Prop}
\begin{proof} For $n=3, \nu=\mu+1/2$ the matrices $R_i$ \eqref{Ri} become
$$
R_i=\mu\left(\begin{array}{ccc}
0	&\psi_{i2}\psi_{i3}	&2\psi_{i3}^2\\
0	&\psi_{i2}^2		&2\psi_{i2}\psi_{i3}\\
0	&\psi_{i1}\psi_{i2}	&2\psi_{i1}\psi_{i3}
\end{array}\right)
$$
Thus the first component decouples from the system and we are left with a 2-dimensional system for the second and third component, which (after the opportune rescaling) coincides with \eqref{2x2new}. The proposition is proved.
\end{proof}
\begin{Rem} The $\nu$-twisted periods of a semisimple Frobenius manifold $M$ can be characterized as the functions on $M$ that have diagonal covariant Hessian in the canonical coordinates, and are homogeneous of degree $\nu+1/2-d/2$; in particular, the flat coordinates $t^\alpha$ are twisted periods with twisting parameter $1/2-\mu_\alpha$. For example, for $n=3$ the degrees of the flat coordinates are 
\begin{equation}\label{degrees}
\text{deg}\,t^1=1\qquad\text{deg}\,t^2=1+\mu\qquad\text{deg}\,t^3=1+2\mu
\end{equation}
where $\mu$ is related to the charge by $\mu=-d/2$. For $\nu=1/2$ the degree of the twisted periods is $1+\mu$, and the second component of \eqref{FuchSys} decouples (see \eqref{t2decoupling}), leading to \eqref{2x21}. For $\nu=\pm\mu+1/2$ the degree coincides with that of $t^3, t^1$ respectively, again producing 2x2 reduced systems, the first of which appears in our construction. The corresponding Painlev\'e transcendents are all related by Okamoto transformations.
\end{Rem}
At this stage combining Theorem 6 and Proposition 3 completes the proof of Theorem 2. The expression of the middle matrix in \eqref{4dtrans} reflects the specific form of the gauge transformations appearing in the proof of Theorem 6.\\
We conclude this section with some remarks on the relation between the normalization of $(\chi_{(1)}, \chi_{(2)})$ and the 4-dimensional metric. Let
$$w\doteq\text{det}(\chi_{(1)}, \chi_{(2)})
$$
be the Wronskian of the solution. From the identities
$$
\frac{\partial}{\partial\epsilon}\log w=\text{Tr}\left(\frac{B_1}{\epsilon}+\frac{B_2}{\epsilon-1}+\frac{B_3}{\epsilon-s}\right)\qquad\frac{\partial}{\partial s}\log w=\text{Tr}\left(-\frac{B_3}{\epsilon-s}\right)
$$
we obtain
\begin{equation}\label{wronsk}
w=\text{const}\,\epsilon^\mu(\epsilon-1)^\mu(\epsilon-s)^\mu
\end{equation}
Then from \eqref{4dtrans} and the normalizations \eqref{norm}  one easily finds
$$
\eta=\hat{\Psi}^t\hat{\Psi}=\text{const}\left(\begin{array}{cccc}
0	&0	&0	&-1\\
0	&0	&1	&0\\
0	&1	&0	&0\\
-1	&0	&0	&0
\end{array}\right)
$$
with the same constant of \eqref{wronsk}.
\section{Examples: Hurwitz spaces}
Hurwitz spaces represent one of the richest class of examples of semi-simple Frobenius manifolds. While the arising of Frobenius structure on Hurwitz spaces is primarily related to their role in the theory of integrable system, as the moduli spaces of $g$-phase solutions of certain integrable hierarchies (\cite{Krich}, \cite{Dub-Nov}), it turns out to have deep connections with a number of related topics (e.g. topological Landau-Ginzburg models, singularity theory, orbit spaces of Coxeter groups and their extensions).\\
Hurwitz spaces are moduli spaces of ramified coverings of the projective line of fixed genus, degree and ramification profile over the point $\infty\in\mathbb{P}^1$. They are especially well-suited to our present context for two reasons: first, tri-hamiltonian structures appear naturally within this class of examples, and there is a simple criterion to recognize them. Second,  the task of computing explicitly the twisted periods of a Frobenius manifold appears very hard in its full generality, but for Hurwitz spaces we have an  somewhat manageable expression for them in terms of certain period integrals on the spectral curve, which arise in the framework of Givental's twisted Picard-Lefschetz theory (\cite{GivTwisted}) (in fact, this is the reason for the name ``twisted periods" itself).\\
\indent We briefly review the main definitions (for a complete discussion of Frobenius structure on Hurwitz spaces, see \cite{Dub2Dtft}, \cite{DubMod}). Let $\rho=(r_1, \dots, r_L)$ be a set of positive integers and $D$ be their sum. A \emph{Hurwitz cover} (or \emph{spectral curve}) \emph{of type} $g, \rho$ is a pair $(C_g, \lambda)$, where
\begin{itemize}
\item $C_g$ is a smooth genus $g$ curve with $L$ marked points $x_1, \dots, x_L$ and a marked canonical homology basis $\{a_k, b_k\}_{k=1, \dots, g}$.
\item $\lambda$ is a degree $D$ meromorphic function on $C_g$ whose branch points are all simple except for poles, which occur precisely at the points $x_i$ with multiplicity $r_i$, $i=1, \dots, L$. We refer to the map $\lambda$ as the \emph{superpotential}.
\end{itemize}
The \emph{Hurwitz space} $\mathcal{M}_{g, \rho}$ is the moduli space of Hurwitz covers of type $g, \rho$, where the equivalence relation is given by isomorphisms of curves compatible with $\lambda$. Its dimension is $n=2g-2+d+L$ and the simple critical values $u_i=\lambda(P_i)\;, d\lambda(P_i)=0,\; i=1, \dots, n$ provide local coordinates on it. These will be the canonical coordinates in the Frobenius structure. We further assume that none of them is zero, according to our general rule of working outside of the discriminant of the Frobenius manifold. Here this means that we exclude from the Hurwitz space the divisor where the zeros of the superpotential are not all simple.\\
To construct the metric we need to select a certain moduli-dependent 1-form on the spectral curve, called the \emph{quasi-momentum differential} and denoted by $\phi$. We will not reproduce here the precise definition and the freedom in the choice of such differential (see \cite{Dub2Dtft}, \cite{DubMod} for a complete list), limiting ourselves to specifying our $\phi$ case by case. We just mention that typically $\phi$ will be either an holomorphic differential or an Abelian differential with poles at the marked points, and it will always be normalized to have constant (i.e. moduli-independent) $a$-periods.\\
Given $\phi$, we define
\begin{equation}\label{reseta}
\eta=\sum_{i=1}^n\eta_idu_i^2\qquad\eta_i\doteq\text{Res}_{P_i}\frac{\phi^2}{d\lambda}
\end{equation}
(recall that $P_i$ denotes the branch point corresponding to the critical value $u_i$). Then the metric is flat and, together with the standard formulas \eqref{prod}, \eqref{eE}, defines a semisimple Frobenius structure on $\mathcal{M}_{g, \rho}$.\\
We now provide precise statements for the facts mentioned at the beginning of the section:
\begin{Thm} A Hurwitz space $\mathcal{M}_{g, \rho}$ of dimension at least 4 has tri-hamiltonian structure if and only if $\rho=(1, \dots, 1)$, i.e. all poles of the superpotential are simple.
\end{Thm}
\begin{proof} It suffices to look at the homogeneity degrees of the flat coordinates of $\eta$ (which are given explicitely in \cite{Dub2Dtft}) and check that they split into two blocks of the same degree exactly when $\rho=(1, \dots, 1)$. The statement holds for any admissible choice of quasi-momentum differential.
\end{proof}
For the next result we need to fix some notation (see Givental's paper \cite{GivTwisted} for details). Let $(C_g, \lambda)$ be a point in the moduli space, and $y_1, \dots, y_d$ be the zeros of $\lambda$. For every value of $\nu$, consider the (typically $\infty$-sheeted) covering of $\hat{C}_g\equiv C_g\setminus\{x_i, y_j\}$ where the function $\lambda^\nu$ is defined. If $q=e^{2\pi i\nu}$, it is described by the local system $L(q)$ on $\hat{C}_g$ whose fiber is $\mathbb{Z}[q, q^{-1}]$ and whose monodromy is multiplication by $q$ around $y_j$ and by $q^{-m_i}$ around $x_i$. 
\begin{Thm} Let $H_1(\hat{C}_g, \{y_j\}; L(q))$ be the first  homology group of $\hat{C}_g$ with coefficients in the local system $L(q)$, relative to (a tubular neighborhood of) the points $\{y_j\}$. Then, for all $\gamma\in H_1(\hat{C}_g, \{y_j\}; L(q))$, the function
$$
p_\gamma(u;\nu)=\int_\gamma\lambda^\nu\phi
$$
is a twisted period for the Frobenius structure on $\mathcal{M}_{g, \rho}$ with quasi-momentum differential $\phi$.
\end{Thm}
The proof requires some technical machinery of Frobenius structures on Hurwitz spaces, and we will omit it here. Let us remark though that we will only use this theorem for the case of $\mathcal{M}_{0, (4)}$, which coincides (as a Frobenius manifold) with the universal unfolding of the simple singularity of type $A_3$. For Frobenius manifolds associated to simple singularities, the above result can be found in \cite{DubAlmost*}.
\begin{Rem}In practice a basis of twisted periods can be constructed as follows: near each critical point $P_i$ of $\lambda$ fix a vanishing 0-cycle, i.e. a pair of points in the fiber of $\lambda$  meeting at the critical point. Moving these two points along the fibers of $\lambda$ from the critical point to $\lambda^{-1}(0)$ defines a relative 1-cycle, which can then be lifted arbitrarily to the $\infty$-sheeted covering of $\hat{C}_g$ and regarded as a relative cycle with local coefficients. Doing this for each critical point one obtains a basis of $H_1(\hat{C}_g, \{y_j\}; L(q))$ (which is easily seen to have dimension $n$).
\end{Rem}
\indent Let us look specifically at the low-dimensional cases relevant to our context. According to Theorem 7, there are two 4-dimensional Hurwitz spaces with tri-hamiltonian structure: $\mathcal{M}_{0, (1, 1, 1)}$ and $\mathcal{M}_{1, (1, 1)}$. Consider in these two spaces the ``boundary divisor" obtained by letting one of the canonical coordinates tend to $\infty$, i.e. by letting two simple poles of the superpotential merge into a double pole; it is clear that the results are respectively the 3-dimensional Hurwitz spaces $\mathcal{M}_{0, (1, 2)}$ and $\mathcal{M}_{1, (2)}$. We intuitively expect these two pairs of Frobenius manifolds to be related by the correspondence of the previous section: in other words, our construction should correspond to the rather trivial operation of ``splitting a double pole into two simple poles'' in the case of Hurwitz spaces with all poles of the superpotential simple except one which is double. We prove this explicitly for the pair of genus one spaces in the first example below.\\
Now, in addition to $\mathcal{M}_{0, (1, 2)}$ and $\mathcal{M}_{1, (2)}$, there is a third 3-dimensional Hurwitz space, namely $\mathcal{M}_{0, (4)}$. In this case we know a priori that the associated 4-dimensional tri-hamiltonian manifold will not be a Hurwitz space, and we lose the previous interpretation of the $n=3\to n=4$ map. In the second example we compute the solutions of the Fuchisan system \eqref{2x2new} for the $A_3$ singularity and find an explicit algebraic expression of the 4-dimensional matrix $\hat{\Psi}$ as function of the canonical coordinates. This will yield a reasonable guess for the corresponding tri-hamiltonian prepotential.
\subsection{Genus one double coverings of $\mathbb{P}^1$} In this example we compute the solutions of Darboux-Egoroff for the Hurwitz spaces $\mathcal{M}_{1, (2)}$ and $\mathcal{M}_{1, (1, 1)}$ and show that they are related as in Proposition 2. As a by-product we obtain a nice procedure to reduce certain elliptic integrals involving the square root of a degree 4  polynomial to Weierstrass-type integrals.\\
Let us begin with $\mathcal{M}_{1, (2)}$. By definition it parametrizes hyperelliptic curves with a branch point at infinity:
\begin{equation}\label{3bp}
\rho^2=(\lambda-u_1)(\lambda-u_2)(\lambda-u_3)
\end{equation}
Using the Weierstrass uniformization,
\begin{equation}\label{Psup}
\lambda=\wp(z; \omega_1, \omega_2)+c\qquad \rho=\frac{1}{2}\frac{d\lambda}{dz}
\end{equation}
we identify the Hurwitz space with the family of elliptic functions $\lambda(z;\omega_1, \omega_2, c)$, parametrized by the periods of the Weierstrass $\wp$ function and the additive constant. Let the normalized holomorphic differential
$$
\phi\doteq \frac{dz}{\omega_1}=\frac{d\lambda}{2\omega_1\rho}
$$
be the quasi-momentum. From \eqref{reseta} we obtain 
\begin{equation}\label{3deta}
\eta_i=\frac{1}{2\omega_1u_{ij}u_{ik}}\qquad i=1, 2, 3
\end{equation}
where the indices $i, j, k$ are understood to be distinct. On the submanifold $u_1=0, u_2=1, u_3=s$, \eqref{3bp} and \eqref{Psup} reduce to
\begin{equation}\label{3bpnorm}
\bar{\rho}^2=\bar{\lambda}(\bar{\lambda}-1)(\bar{\lambda}-s)
\end{equation}
$$
\bar{\lambda}=\wp(z;\bar{\omega}_1, \bar{\omega}_2)+\frac{s+1}{3}\qquad\bar{\rho}=\frac{1}{2}\frac{d\bar{\lambda}}{dz}
$$
\begin{Lemma} The matrix $V$ for the metric \eqref{3deta} has the form
\begin{equation}\label{ellipticV}
V(u_1, u_2, u_3)=\left(\begin{array}{ccc}
0	&-c(s)	&b(s)\\
c(s)	&0		&-a(s)\\
-b(s)	&a(s)	&0
\end{array}\right)
\end{equation}
$$
a(s)=\frac{1}{2\sqrt{-s}}\bar{I}(s)\qquad b(s)=-\frac{1}{2\sqrt{s-1}}(\bar{I}(s)-1)\qquad c(s)=\frac{1}{2\sqrt{s(1-s)}}(\bar{I}(s)-s)
$$
$$
\bar{I}(s)=\frac{s+1}{3}-\frac{2\bar{\eta}_1}{\bar{\omega}_1}
$$
where $\bar{\eta}_1=\zeta(\bar{\omega}_1/2;\bar{\omega}_1, \bar{\omega}_2)$ and $\zeta$ is the Weierstrass zeta-function on the curve \eqref{3bpnorm}
\end{Lemma}
\begin{proof}
 To compute the rotation coefficients we need to differentiate the period $\omega_1$ with respect to the branch points. To this end, realize the curve \eqref{3bp} by gluing two copies of the complex plane along cuts drawn from $u_1$ to $u_2$ and from $u_3$ to $\infty$. Thus
$$
\omega_1(u_1, u_2, u_3)=\oint_a\frac{d\lambda}{2\rho}=\oint_a\frac{d\lambda}{2\sqrt{(\lambda-u_1)(\lambda-u_2)(\lambda-u_3)}}
$$
where $a$ is a contour in the first sheet encircling the first cut. Introduce the elliptic integrals
\begin{equation}\label{Iidef}
I_i\doteq2\frac{1}{\omega_1}\frac{\partial\omega_1}{\partial u_i}=\oint_a\frac{\phi}{\lambda-u_i}\qquad i=1, 2, 3
\end{equation}
Now, $(\lambda-u_i)^{-1}$ and $\lambda(z-\omega_i/2)-u_i$ are elliptic functions on the curve \eqref{3bp} with the same zeroes and poles. Confronting their expansion at $\omega_i/2$ we find
$$
\frac{1}{\lambda-u_i}=\frac{1}{u_{ij}u_{ik}}\big(\lambda(z-\omega_i/2)-u_i\big)
$$
hence
\begin{equation}\label{Ii}
I_i(u_1, u_2, u_3)=\frac{1}{u_{ij}u_{ik}}\oint_a(\lambda-u_i)\phi
\end{equation}
After the change of variable $\bar{\lambda}=(\lambda-u_1)/u_{21}$ , we obtain
\begin{equation}\label{Ibar}
I_1=\frac{\bar{I}(s)}{u_{21}s}\qquad I_2=\frac{\bar{I}(s)-1}{u_{21}(1-s)}\qquad I_3=\frac{\bar{I}(s)-s}{u_{21}s(s-1)}
\end{equation}
where
$$
\bar{I}(s)=\frac{1}{\bar{\omega}_1}\oint_a\frac{\bar{\lambda} d\bar{\lambda}}{2\bar{\rho}}=\frac{s+1}{3}-\frac{2\bar{\eta}_1}{\bar{\omega}_1}
$$
At this point a straightforward computation using the definition \eqref{gamma}, \eqref{defV} completes the proof.
\end{proof}
We now repeat the same computation for the case $\mathcal{M}_{1, (1, 1)}$ of 4 finite branch points, i.e.
\begin{equation}\label{4bp}
\rho^2=(\lambda-v_1)(\lambda-v_2)(\lambda-v_3)(\lambda-v_4)
\end{equation}
We use the uniformization
\begin{equation}\label{Zsup}
\lambda=\zeta(z-x;\Omega_1, \Omega_2)-\zeta(z+x;\Omega_1, \Omega_2)+c\qquad \rho=\frac{d\lambda}{dz}
\end{equation}
in terms of the Weierstrass $\zeta$ function. Here, to avoid confusion with the previous case, we denoted the periods by capital letters $\Omega_i$. The additional coordinate $x$ describes the position of the poles of $\lambda$ in the uniform parameter $z$. Choosing as before the normalized holomorphic differential
$$
\phi=\frac{dz}{\Omega_1}=\frac{d\lambda}{\Omega_1\rho}
$$
as the quasi-momentum differential, we find
\begin{equation}\label{4deta}
\eta_i=\frac{2}{\Omega_1^2v_{ij}v_{ik}v_{il}}\qquad i=1, 2, 3, 4
\end{equation}
On the submanifold $v_1=0, v_2=1$ we have
\begin{equation}\label{4bpnorm}
\bar{\rho}^2=\bar{\lambda}(\bar{\lambda}-1)(\bar{\lambda}-P)(\bar{\lambda}-Q)
\end{equation}
$$
P=\frac{v_{31}}{v_{21}}=s\frac{\epsilon-1}{\epsilon-s}\qquad Q=\frac{v_{41}}{v_{21}}=1-\epsilon
$$
in terms of the parameters $s, \epsilon$ of \eqref{s, eps}. To compute the rotation coefficients of the metric \eqref{4deta} and compare them with \eqref{ellipticV}, the main step is to reduce the relevant elliptic integrals, involving the square root of a degree 4 polynomial, to elliptic integrals involving the square root of a degree 3 polynomial. A useful formula in this regard is the following:
\begin{Lemma} The identity
$$
\partial_{\tilde{e}}\Omega_1=-\frac{1}{2}\sum_{i=1}^4v_i\,\Omega_1
$$
holds true.
\end{Lemma}
\begin{proof} Define, as in \eqref{Iidef}
\begin{equation}\label{Ji}
J_i=2\frac{1}{\Omega_1}\frac{\partial}{\partial v_i}\Omega_1=\oint_a\frac{\phi}{\lambda-v_i}\qquad i=1, 2, 3, 4
\end{equation}
so that
\begin{equation}\label{conto}
\partial_{\tilde{e}}\Omega_1=\frac{1}{2}\sum_{i=1}^4v_i^2J_i\,\Omega_1=\frac{1}{2}\oint_a\sum_{i=1}^4\frac{v_i^2}{\lambda-v_i}\phi\;\Omega_1
\end{equation}
Comparing zeros and poles and expanding at $\Omega_i/2$ one checks the identity
$$
\frac{1}{\lambda-v_i}=\frac{4}{v_{ij}v_{ik}v_{il}}\wp(z-\Omega_i/2)-\frac{1}{3}\left(\frac{1}{v_{ij}}+\frac{1}{v_{ik}}+\frac{1}{v_{il}}\right)
$$
Substituting in \eqref{conto} we obtain the Lemma.
\end{proof}
Let us explain how to use the lemma to compare for example
$$
\bar{J}_1(\epsilon, s)\equiv\frac{1}{v_{21}}J_1=\frac{1}{\bar{\Omega}_1}\oint_a\frac{d\lambda}{\lambda\sqrt{\lambda(\lambda-1)(\lambda-P)(\lambda-Q)}}
$$
and
$$
\bar{I}_1(s)\equiv\frac{1}{u_{21}}I_1=\frac{1}{\bar{\omega}_1}\oint_a\frac{d\lambda}{2\lambda\sqrt{\lambda(\lambda-1)(\lambda-s)}}
$$
where
$$
\bar{\Omega}_1=\oint_a\frac{d\lambda}{\sqrt{\lambda(\lambda-1)(\lambda-P)(\lambda-Q)}}\qquad\bar{\omega}_1=\oint_a\frac{d\lambda}{2\sqrt{\lambda(\lambda-1)(\lambda-s)}}
$$
(recall that $P, Q$ are functions of $\epsilon, s$ as in \eqref{4bpnorm}). Using \eqref{Ji} and the lemma, compute
$$
\partial_{\tilde{e}}J_1=2\partial_{\tilde{e}}\frac{\partial_1\Omega_1}{\Omega_1}=\frac{2}{\Omega_1}\left(\partial_1\partial_{\tilde{e}}\Omega_1-2v_1\partial_1\Omega_1-\frac{\partial_1\Omega_1}{\Omega_1}\partial_{\tilde{e}}\Omega_1\right)=-1-2v_1J_1
$$
Substituting $J_1=v_{21}\bar{J}_1(\epsilon, s)$ the last equation becomes
$$
\frac{\partial}{\partial\epsilon}\bar{J}_1=\frac{1}{\epsilon(\epsilon-1)}(1-\bar{J}_1)
$$
which implies
\begin{equation}\label{barJ1}
\bar{J}_1=\frac{1}{\epsilon-1}(\epsilon\bar{I}_1(s)-1)
\end{equation}
because $\bar{J}_1(\epsilon, s)\to\bar{I}_1(s)$ for $\epsilon\to\infty$.
\begin{Prop} The matrix $W$ for the metric \eqref{4deta} has the form
$$
W(v_1, v_2, v_3, v_4)=\left(\begin{array}{cccc}
0	&-c(s)	&b(s)	&-a(s)\\
c(s)	&0		&-a(s)	&-b(s)\\
-b(s)	&a(s)	&0		&-c(s)\\
a(s)	&b(s)	&c(s)		&0
\end{array}\right)\qquad s=\frac{v_{31}v_{42}}{v_{21}v_{43}}
$$
where $a, b, c$ are the same of Lemma 4.
\end{Prop}
\begin{proof}
Let us check the statement for the component $W_{12}$:
$$
W_{12}=v_{21}\frac{1}{\sqrt{\eta_1}}\frac{\partial\sqrt{\eta_2}}{\partial v_1}=\frac{1}{2}v_{21}\sqrt{\frac{v_{12}v_{13}v_{14}}{v_{21}v_{23}v_{24}}}\left(\frac{1}{v_{21}}-J_1\right)=\frac{1}{2}\sqrt{\frac{s}{1-s}}\frac{\epsilon-1}{\epsilon}(1-\bar{J}_1)
$$
Plugging \eqref{barJ1} into the last expression we obtain
$$
W_{12}=\frac{1}{2}\sqrt{\frac{s}{1-s}}(1-\bar{I}_1(s))=-c(s)
$$
since $\bar{I}_1(s)=\bar{I}(s)/s$ (see \eqref{Ibar}). The equality of the other components is checked in the same way.
\end{proof}
\subsection{The $A_3$ singularity}
Let us now consider the case of $\mathcal{M}_{0, (4)}$. After a suitable uniformization of spectral curves, we identify the Hurwitz space with the family of polynomials
\begin{equation}\label{A3sup}
\lambda(z; t^1, t^2, t^3)=z^4+4t^3 z^2+2t^2z+t^1+2(t^3)^2
\end{equation}
which coincides with the universal unfolding of the simple singularity of type $A_3$,  $\lambda(z; 0, 0, 0)=z^4$. Let
$$
\phi= dz
$$
be the quasi-momentum differential. Then $t^1, t^2, t^3$ are flat coordinates of the metric \eqref{reseta} and reduce it to antidiagonal form (see \cite{Dub2Dtft}). Their degree of homogeneity is immediately read off \eqref{A3sup}: since the degree of $\lambda$ is 1 (the same of the canonical coordinates), we have deg$\;z=1/4$ and
$$
\text{deg}\;t^1=1\qquad\text{deg}\;t^2=3/4\qquad\text{deg}\;t^3=1/2
$$
In particular, $\mu=-1/4$ (see \eqref{degrees}).\\
In order to construct the associated 4-dimensional manifold, the first step is to compute the rescaled transition matrix $\Phi(s)$:
\begin{Lemma} The rescaled transition matrix $\Phi$ for $\mathcal{M}_{0, (4)}$ is the following algebraic function of $s$:
$$
\Phi(s)=\left(\begin{array}{ccc}
\frac{(1+2t)^{3/4}}{2\sqrt{(2+t)(1+2t)}}	&\frac{-1-t}{\sqrt{(2+t)(1+2t)}}	&\frac{1+3t+t^2}{\sqrt{2+t}(1+2t)^{5/4}}\\
\frac{(1+2t)^{3/4}}{2\sqrt{(-1+t)(1+2t)}}	&\frac{t}{\sqrt{(-1+t)(1+2t)}}		&\frac{-1-t+t^2}{\sqrt{-1+t}(1+2t)^{5/4}}\\
\frac{(1+2t)^{3/4}}{2\sqrt{(1-t)(2+t)}}	&\frac{1}{\sqrt{(1-t)(2+t)}}		&\frac{1-t-t^2}{\sqrt{(1-t)(2+t)}(1+2t)^{3/4}}
\end{array}\right)
$$
where
$$
s=\frac{t(2+t)^3}{(1+2t)^3}
$$
\end{Lemma}
\begin{proof}
We introduce a new parametrization of the Hurwitz space, choosing as coordinates two of the critical points of $\lambda$ (the sum of all three is zero) plus $t^1$:
$$
\lambda'(z)=4(z-q_2)(z- q_3)(z+q_2+q_3)
$$
$$
\lambda(z)= z^4-2(q_2^2+q_2q_3+q_3^2)z^2+4q_2q_3(q_2+q_3)z+t^1+\frac{1}{2}(q_2^2+q_2q_3+q_3^2)^2
$$
The advantage of this parametrization is that both the flat and the canonical coordinates are polynomials in $t^1, q_2, q_3$:
$$
t^2=2q_2q_3(q_2+q_3)\qquad t^3=-\frac{1}{2}(q_2^2+q_2q_3+q_3^2)
$$
\begin{align*}
u_1	&=\lambda(-q_2-q_3)=-\frac{1}{2}(q_2^4+10q_2^3q_3+17q_2^2q_3^2+10q_2q_3^2+q_3^4)+t^1\\
u_2	&=\lambda(q_2)=\frac{1}{2}(-q_2^4+6q_2^3q_3+7q_2^2q_3^2+2q_2q_3^3+q_3^4)+t^1\\
u_3	&=\lambda(q_3)=\frac{1}{2}(q_2^4+2q_2^3q_3+7q_2^2q_3^2+6q_2q_3^3-q_3^4)+t^1
\end{align*}
Now we can easily compute the Jacobian $J=\{\partial u_i/\partial t^\alpha\}$
$$
J=\left(\begin{array}{ccc}
1	&-2(q_2+q_3)	&2(q_2^2+3q_2q_3+q_3^2)\\
1	&2q_2			&2(q_2^2-q_2q_3-q_3^2)\\
1	&2q_3			&2(-q_2^2-q_2q_3+q_3^2)
\end{array}\right)
$$
and the diagonal coefficients of the metric \eqref{reseta}
\begin{align*}
\eta_1	&=\text{Res}_{z=-(q_2+q_3)}\frac{dz}{\lambda'(z)}=\frac{1}{4(2q_2+q_3)(q_2+2q_3)}\\
\eta_2	&=\text{Res}_{z=q_2}\frac{dz}{\lambda'(z)}=\frac{1}{4(q_2-q_3)(2q_2+q_3)}\\
\eta_3	&=\text{Res}_{z=q_3}\frac{dz}{\lambda'(z)}=\frac{1}{4(q_3-q_2)(q_2+2q_3)}
\end{align*}
At this point the formula
$$
\psi_{i\alpha}=\sqrt{\eta_i}\frac{\partial u_i}{\partial t^\alpha}
$$
yields the fundamental solution
$$
\Psi=\left(\begin{array}{ccc}
\frac{1}{2\sqrt{(2q_2+q_3)(q_2+2q_3)}}	&\frac{-q_2-q_3}{\sqrt{(2q_2+q_3)(q_2+2q_3)}}	&\frac{q_2^2+3q_2q_3+q_3^2}{\sqrt{(2q_2+q_3)(q_2+2q_3)}}\\
\frac{1}{2\sqrt{(q_2-q_3)(2q_2+q_3)}}	&\frac{q_2}{\sqrt{(q_2-q_3)(2q_2+q_3)}}	&\frac{q_2^2-q_2q_3-q_3^2}{\sqrt{(q_2-q_3)(2q_2+q_3)}}\\
\frac{1}{2\sqrt{(q_3-q_2)(q_2+2q_3)}}	&\frac{q_3}{\sqrt{(q_3-q_2)(q_2+2q_3)}}	&\frac{-q_2^2-q_2q_3+q_3^2}{\sqrt{(q_3-q_2)(q_2+2q_3)}}
\end{array}\right)
$$
Finally, we set $t\doteq q_2/q_3$, so that $s=u_{31}/u_{21}=t(2+t)^3(1+2t)^{-3}$. After the rescaling $\Phi=\Psi u_{21}^{-\hat{\mu}}$ with $\hat{\mu}=\text{diag}\;(-1/4, 0, 1/4)$, we arrive at the final expression.
\end{proof}
 The next step is to compute the solutions of the Fuchsian system \eqref{2x2new}. Let us first consider twisted periods in general for a moment. Let $\xi_1,\; \xi_2,\; \xi_3,\; \xi_4=-(\xi_1+\xi_2+\xi_3)$ be the roots of $\lambda$ (which are all distinct due to our assumption $u_i\neq 0$) and $\gamma_i, \;i=1, 2, 3$ be paths in the complex plane from $\xi_i$ to $\xi_{i+1}$. For each $\nu$, we lift $\gamma_i$ to the covering of $\mathbb{C}\setminus\{\xi_i\}$ where $\lambda^\nu$ is defined and obtain a basis of $H_1(\mathbb{C}\setminus\{\xi_i\}, \{\xi_i\}, L(q))$ (see Remark 5). We conclude that a system of twisted periods is given by the formula
\begin{equation}\label{A3periods}
p_i(t, \nu)=\int_{\xi_i}^{\xi_{i+1}}\prod_{j=1}^4(z-\xi_j)^\nu dz\qquad i=1, 2, 3
\end{equation}
where the integration is meant on the lifting of $\gamma_i$, i.e. for some fixed branch of the integrand. Using the integral representation of the Appell two-variables hypergeometric function $F_1$ (\cite{App}),
\begin{align}\label{AppellF1}
F_1(\alpha;\beta, \beta';\gamma;x, y)	&\equiv\sum_{m, n \geq 0}\frac{(\alpha)_{m+n}(\beta)_m(\beta')_n}{(\gamma)_{m+n}m!n!}x^my^n=\\
						&=\frac{\Gamma(\gamma)}{\Gamma(\alpha)\Gamma(\gamma-\alpha)}\int_0^1\frac{u^{\alpha-1}(1-\alpha)^{\gamma-\alpha-1}}{(1-xu)^\beta(1-yu)^{\beta'}}du\nonumber
\end{align}
valid for $\text{Re}(\alpha)>\text{Re}(\gamma-\alpha)>0$, we can express the \eqref{A3periods} in the form
\begin{equation}\label{Appellperiods}
p_i(t, \nu)=\xi_{(i+1)i}^{2\nu+1}[\xi_{(i+2)i}\xi_{i(i+3)}]^\nu\frac{\Gamma(\nu+1)^2}{\Gamma(2\nu+2)}F_1(\nu+1;-\nu, -\nu;2\nu+2; x_i, y_i)
\end{equation}
$$
x_i=\frac{\xi_{(i+1)i}}{\xi_{(i+2)i}}\qquad y_i=\frac{\xi_{(i+1)i}}{\xi_{(i+3)i}}\qquad i=1, 2, 3
$$
where as usual $\xi_{ij}=\xi_i-\xi_j$ and the indexes are cyclically ordered. Since
\begin{equation}\label{partialp}
\frac{\partial}{\partial t^\alpha}p(t;\nu)=\nu\int_\gamma\lambda^{\nu-1}\frac{\partial\lambda}{\partial t^\alpha}dz
\end{equation}
$$
\frac{\partial\lambda}{\partial t^1}=1\qquad\frac{\partial\lambda}{\partial t^2}= 2z\qquad \frac{\partial\lambda}{\partial t^3}=4(z^2+t^3)
$$
one can apply the same method and obtain expressions similar to \eqref{Appellperiods} for the solutions of the Fuchsian system \eqref{FuchSys}.\\
\indent Let us now specialize to the case $\nu=\mu+1/2=1/4$. Here the $\infty$-sheeted covering of $\mathbb{C}\setminus\{\xi_i\}$ that we had for general $\nu$ reduces to the degree 4 cyclic covering
\begin{equation}\label{g3curve}
w^4=\lambda(z, t)=(z-\xi_1)(z-\xi_2)(z-\xi_4)(z-\xi_4)
\end{equation}
which is a curve of genus 3, and the picture simplifies significantly. First we note that, since the $\gamma_i$ live in the homology relative to $\{\xi_i\}$, their sum is equivalent to a loop around $\infty$ and we have
$$
p_1(t;1/4)+p_2(t;1/4)+p_3(t;1/4)=-2\pi i\text{Res}_{z=\infty}wdz=2\pi it^3
$$
(up to a 4th root of 1 factor), in agreement with the fact that $t^3$ is a twisted period. We can then use $p_1, p_2$ to construct a fundamental solution of \eqref{2x2new}:
\begin{equation}\label{chi(i)}
\chi_{(i)}=\left(\begin{array}{c}\frac{\partial}{\partial t^2}p_i\\\frac{\partial}{\partial t^1}p_i\end{array}\right)=\frac{1}{4}\left(\begin{array}{c}2\int_{\gamma_i}\omega_2\\\int_{\gamma_i}\omega_1\end{array}\right)\qquad i=1, 2
\end{equation}
where
$$
\omega_1=4\frac{\partial w}{\partial t^1}dz=\frac{dz}{w^3}\qquad\omega_2=2\frac{\partial w}{\partial t^2}dz=\frac{zdz}{w^3}
$$
which are holomorphic differentials on the curve \eqref{g3curve}.
\begin{Rem} The choice of the paths $\gamma_i$ in the formulas \eqref{chi(i)} is not essential: the periods of $\omega_1, \omega_2$ along any cycle on the genus 3 curve will produce a solution of the Fuchsian system. This does not contradict the fact that the space of solutions of the system has dimension 2: indeed, choosing a basis of $a$ and $b$ cycles on \eqref{g3curve} in a way compatible with the order 4 automorphism
$$
\sigma:(z, w)\mapsto(z, iw)
$$
and using the fact that
$$
\sigma^*\omega_k=i\omega_k\qquad k=1, 2
$$
one can check that every period of $\omega_k$ is a linear combination of the integrals over $\gamma_1$ and $\gamma_2$.
\end{Rem}
\begin{Prop} Let $\xi_i=\xi_i(\epsilon, s),\; i=1, 2, 3, 4$ be the roots of the polynomial
\begin{equation}\label{rescsup}
\bar{\lambda}(z; \epsilon, s)=z^4-2\frac{1+t+t^2}{(1+2t)^{3/2}}z^2+4\frac{t(1+t)}{(1+2t)^{9/4}}z+\frac{(1+t)^2(1+4t+t^2)}{(1+2t)^3}-\epsilon
\end{equation}
where $s=t(2+t)^3(1+2t)^{-3}$ as in Lemma 6. Define
$$
\chi_{(i)}=\frac{1}{\sqrt{\xi_{(i+1)i}}[\xi_{(i+2)i}\xi_{i(i+3)}]^{3/4}}\left(\begin{array}{c}\xi_{(i+1)i}f(x_i, y_i)+2\xi_ig(x_i, y_i)\\g(x_i, y_i)\end{array}\right)\qquad i=1, 2
$$
where the notation is the same of \eqref{Appellperiods}, and $f, g$ are the functions
$$
f(x, y)=\frac{\sqrt{\sqrt{1-y}+\sqrt{x-y}}-\sqrt{\sqrt{1-y}-\sqrt{x-y}}}{\sqrt{(1-x)(1-y)(x-y)}}
$$
$$
g(x, y)=\frac{1}{\sqrt{2}}\left(1+\frac{1}{\sqrt{(1-x)(1-y)}}\right)\sqrt{\frac{1}{\sqrt{1-x}+\sqrt{1-y}}}
$$
Then $(\chi_{(1)}, \chi_{(2)})$ is a fundamental solution of the Fuchsian system \eqref{2x2new} for the Hurwitz space $\mathcal{M}_{0(4)}$.
\end{Prop}
\begin{proof}
The $(\epsilon, s)$-parametric family of polynomials \eqref{rescsup} is nothing but the submanifold $\{u_1=-\epsilon, u_2=1-\epsilon, u_3=s-\epsilon\}$ of the whole Hurwitz space, which is where we need to compute \eqref{chi(i)} in order to get the solution of \eqref{2x2new}. The explicit formula follows from a computation similar to \eqref{Appellperiods}, with the new feature that the hypergeometric functions appearing become algebraic for the special values of the parameters, and the final result is expressed in elementary functions. We report below the relevant reduction formulas for Appell functions.\\
The function $g(x, y)$ is in fact $F_1(1/4;3/4, 3/4;1/2;x,y)$. Its algebraic reduction follows from the duplication formula proved in \cite{IsmPit}
$$
F_1(\alpha; \alpha+\frac{1}{2}, \alpha+\frac{1}{2};2\alpha;x, y)=\frac{1}{2}\left(1+\frac{1}{\sqrt{(1-x)(1-y)}}\right)\left(\frac{2}{\sqrt{1-x}+\sqrt{1-y}}\right)^{2\alpha}
$$
taking $\alpha=1/4$. The function $f(x, y)$ is $F_1(5/4;3/4, 3/4;3/2; x, y)$, and can be reduced to algebraic form using the combination of the two identities (see e.g. \cite{EMOT})
$$
F_1(\alpha;\beta, \beta';\beta+\beta';x, y)=(1-y)^{-\alpha} {_{2}F_{1}} (\alpha, \beta; \beta+\beta';\frac{x-y}{1-y})
$$
$$
 {_{2}F_{1}}(a, a+\frac{1}{2};\frac{3}{2}, z^2)=\frac{1}{(2-4a)z}((1+z)^{1-2a}-(1-z)^{1-2a})
$$
where $ {_{2}F_{1}}$ is the usual hypergeometric function.
\end{proof}
At this stage we are in a position to combine Lemma 6 and Proposition 5 into \eqref{4dtrans} and write down an explicit expression of the matrix $\hat{\Psi}$, which would definitely look quite ugly but nevertheless be an algebraic function. An immediate consequence is that the prepotential $F$ of the 4-dimensional manifold is algebraic as well.\\
It is an old conjecture by Dubrovin (\cite{Dub2Dtft}, \cite{DubPainP}, see also \cite{Dinar}) that semisimple algebraic solutions of WDVV with positive degrees correspond to primitive conjugacy classes in Coxeter groups. Remarkably, a 4-dimensional Frobenius manifold that resembles in all aspects the result of our construction for $A_3$ was indeed found by Pavlyk in \cite{Pavlyk}, in correspondence to a primitive conjugacy class in the Weil group of the Lie algebra $D_4$. His algebraic prepotential looks as follows,
\begin{align*}
F(t_1, t_2, t_3, t_4)= &\frac{1}{2}t_1^2t_4+t_1t_2t_3+\frac{1}{6}t_2^2t_4-\frac{1}{108}t_2t_4^3+\frac{1}{12}t_2t_3^2t_4+\frac{19}{2^8\cdot3^4\cdot5}t_4^5+\\
&+\frac{7}{2^7\cdot3^3}t_3^2t_4^3+\frac{1}{3\cdot2^8}t_3^4t_4+\frac{(48t_2+3t_3^2+t_4^2)^{5/2}}{2^5\cdot 3^4\cdot 5}
\end{align*}
(here we lowered the indices of the flat coordinates for notational ease), and it is easily checked to be a tri-hamiltonian solution of WDVV with $\mu=-1/4$.\\
Unfortunately, given the size of the formulas, it practically looks arduous either to derive the prepotential from our transition matrix or to do the converse from Pavlyk's prepotential. Therefore we are presently unable to prove explicitly that the two Frobenius manifolds coincide, and we content ourselves with suggesting that reasonably enough that should be the case.
\section*{Concluding remarks}
An immediate question is how the picture looks like in dimension higher than 4. Actually, it is relatively straightforward to generalize Proposition 2 to a correspondence between $2n-1$ and $2n$ dimensional solutions of Darboux-Egoroff that are ``maximally degenerate", in the sense that the spectrum consists of $\pm\mu$ with maximal multiplicity, plus a single zero eigenvalue in odd dimension. In essence it suffices to note that the space of such matrices in $so(2n-1)$ and in $so(2n)$  has the same dimension $n(n-1)+1$.\\
Thus in principle all the information of a tri-hamiltonian Frobenius manifold can be recovered from the codimension 1 submanifold obtained by sending one of the canonical coordinates to $\infty$. The special feature of dimension 4 is that this procedure yields \emph{all} the 3-dimensional semisimple Frobenius manifolds, while in higher dimension it clearly gives only a small subclass. In this sense the case studied here appears to be the most interesting.\\
Further questions concern the study of the dimension 3/dimension 4 correspondence from the point of view of the associated hydrodynamic-type hierarchies, as well as the clarification of the role of tri-hamiltonianity in this context, and possibly in the higher genera theory of Frobenius manifolds. A suggestive circumstance in this regard is the appearance of $sl(2)$-type Virasoro constraints in unitary matrix models related to the Ablowitz-Ladik hierarchy (\cite{AVM1}, \cite{AVM2}), which in the dispersionless limit match precisely the action of the vector fields $\partial_e, \partial_E, \partial_{\tilde{e}}$ of the underlying tri-hamiltonian Frobenius structure on the dispersionless tau-function.\\
Another very concrete goal would be to apply the reconstruction procedure to other known 3-dimensional Frobenius manifolds/PVI$\mu$ transcendents in hope of explicitly finding new solutions of WDVV in dimension 4.

\vspace{45pt}
\center{Stefano Romano}\\
\center{SISSA -- International School for Advanced Studies}\\
\center{Via Bonomea, 265 - 34136 Trieste ITALY}\\
\center{Email: \nolinkurl{sromano@sissa.it}}

\begin{thebibliography}{50}
\bibitem{AVM1} M. Adler - P. van Moerbeke, \emph{Recursion relations for unitary integrals, combinatorics and the Toeplitz lattice}. Comm. Math. Phys. 237 (2003), 397-440.
\bibitem{AVM2} M. Adler - P. Van Moerbeke, \emph{Integrals over classical groups, random permutations, Toda and Toeplitz lattices}. Comm. Pure Appl. Math. 54 (2001), 153-205.
\bibitem{App} P. Appell, \emph{Sur les fonctions hyperg\'eometriques de plusiers variables}. Mem. des Sciences de Paris, III (1925).
\bibitem{Dinar} Y. Dinar, \emph{On classification and construction of algebraic Frobenius manifolds}. Journal of Geometry and Physics, 58, 9 (2008).
\bibitem{Dub2Dtft} B. Dubrovin, \emph{Geometry of 2D topological field theories}. In: ``Integrable Systems and Quantum Groups" (authors: R. Donagi, B. Dubrovin, E. Frenkel, E. Previato), eds M. Francavilla, S.Greco, Springer Lecture Notes in Math. Vol 1620 (1996), 120-348.
\bibitem{DubPainP} B. Dubrovin, \emph{Painlev\'e transcendents and topological field theory}. In ``The Painlev\'e property: one century later", R. Conte (ed.), Springer Verlag (1999), 287-412.
\bibitem{DubAlmost*}  B. Dubrovin, \emph{On almost duality for Frobenius manifolds}. In ``Geometry, topology and mathematical physics" Amer. Math. Soc. Transl. Ser. 2, 212 (2004), 75-132.
\bibitem{DubFlat} B. Dubrovin, \emph{Flat pencils of metrics and Frobenius manifolds}, Proceedings of 1997 Taniguchi Symposium ``Integrable Systems and Algebraic Geometry" (1998), 47-72.
\bibitem{DubMod} B. Dubrovin, \emph{Differential geometry of moduli spaces and its applications to soliton equations and to topological conformal field theory}. Surveys Diff. Geom. 4 (1999), 213-238.
\bibitem{DubGAT} B. Dubrovin, \emph{Geometry and analytic theory of Frobenius manifolds}, ``Proceedings of the International Congress of Mathematicians (Berlin 1998), Doc. Math. Extra Vol. II (1998),  315-326.
\bibitem{Dub-Maz} B. Dubrovin - M. Mazzocco, \emph{Monodromy of certain Painlev\'e VI transcendents and reflection groups}, Inventiones Math. 141 (2000) 55-147.
\bibitem{Dub-Nov}  B. Dubrovin - S. Novikov, \emph{Hyrdodynamics of weakly deformed soliton lattices. Differential geometry and Hamiltonian theory}. Uspekhi Mat. Nauk 44 (1989) 29-98. English translation in Russ. Math. Surveys 44 (1989), 35-124.
\bibitem{Dub-Nov2}  B. Dubrovin - S. Novikov, \emph{On Poisson brackets of hydrodynamic type}. Soviet Math. Dokl. 279:2 (1984), 294-297.
\bibitem{EMOT} A. Erd\'elyi - W. Magnus - F. Oberhettinger - G. Tricomi, \emph{Higher transcendental functions. Vol. I}, McGraw-Hill (1953).
\bibitem{Gamb} B. Gambier, \emph{Sur les equations differentielles du second ordre et du primier degr\`e
dont l’integrale est a points critiques fixes}, Acta Math. 33, (1910).
\bibitem{GivTwisted}A. Givental, \emph{Twisted Picard-Lefschetz formulas}. Funktsional. Anal. i Prilozhen., 22 (1988), 12-22.
\bibitem{Guzz} D. Guzzetti, \emph{Inverse Problem and Monodromy Data for three-dimensional Frobenius Manifolds}. Math.Phys.Analysis and Geometry 4, (2001), 245-291.
\bibitem{IsmPit}M. Ismail - J. Pitman, \emph{Algebraic evaluation of some Euler integrals, duplication formulae for Appell's hypergeometric function $F_1$, and Brownian variations}, Canadian Journal of Mathematica Physics (2000).
\bibitem{Krich}I. M. Krichever, \emph{Integration of nonlinear equations by the methods of algebraic geometry}. Funct. Anal. Appl. 11 (1977), 12-26.
\bibitem{Man} Yu. I. Manin, \emph{Sixth Painlev\'e equation, universal elliptic curve, and mirror of $\bold{P}^2$}. In ``Geometry of Differential
Equations", Amer. Math. Soc. Transl. 2, 186 (1996), 131--151.
\bibitem{Okam} K. Okamoto, \emph{Studies on the Painlev\'e  equations. I. Sixth Painlev\'e equation PVI}. Ann. Mat. Pura Appl. 146:4 (1987), 337-381.
\bibitem{Painl} P. Painlev\'e, \emph{Sur les equations differentielles du second ordre et d'ordre superieur, dont l’interable generale est uniforme}. Acta Math. 25 (1902).
\bibitem{PT} M. V. Pavlov - S. P. Tsarev, \emph{Tri-Hamiltonian structures of the Egorov systems}. Funct. Anal. and Appl. 37:1 (2003), 32-45.
\bibitem{Pavlyk} O. Pavlyk, \emph{Solutions to WDVV from generalized Drinfeld-Sokolov hierarchies}. Preprint (2003) available at math-ph/0003020.
\bibitem{Sch}L. Schlesinger, \emph{\"{U}ber eine Klasse von Differentsialsystemen beliebliger Ordnung mit festen kritischer Punkten}. J. f\"{u}r Math. 141 (1912), 96-145.
\end{thebibliography}
\end{document}